\documentclass[final,12pt]{article}
\usepackage{amsfonts,color,morefloats,pslatex,a4wide}
\usepackage{amssymb,amsthm,amsmath,latexsym,pslatex,cite}

\newtheorem{theorem}{Theorem}
\newtheorem{lemma}[theorem]{Lemma}

\newtheorem{example}[theorem]{Example}

\usepackage[para]{threeparttable}

\newcommand{\ord}{{\mathrm{ord}}}

\newcommand{\tr}{{\mathrm{Tr}}}

\newcommand{\gf}{{\mathrm{GF}}}

\newcommand{\Aut}{{\mathrm{Aut}}}
\newcommand{\PAut}{{\mathrm{PAut}}}
\newcommand{\MAut}{{\mathrm{MAut}}}

\newcommand{\wt}{{\mathtt{wt}}}

\newcommand{\Z}{\mathbb{{Z}}}

\newcommand{\m}{\mathbb{M}}

\newcommand{\C}{{\mathcal{C}}}

\newcommand{\bc}{{\mathbf{c}}}

\newcommand{\0}{\textbf{0}}

\makeatletter

\newcommand{\Rmnum}[1]{\expandafter\@slowromancap\romannumeral #1@}
\makeatother

\begin{document} 

\title{Several families of ternary negacyclic codes and their duals\thanks{
Z. Sun's research was supported by The National Natural Science Foundation of China under Grant Number 62002093. C. Ding's research was supported by The Hong Kong Research Grants Council, Proj. No. $16301522$}}

\author{Zhonghua Sun\thanks{School of Mathematics, Hefei University of Technology, Hefei, 230601, Anhui, China. Email:  sunzhonghuas@163.com}, 
\and Cunsheng Ding\thanks{Department of Computer Science
                           and Engineering, The Hong Kong University of Science and Technology,
Clear Water Bay, Kowloon, Hong Kong, China. Email: cding@ust.hk}
}

\maketitle

\begin{abstract} 
Constacyclic codes contain cyclic codes as a subclass and have nice algebraic structures. Constacyclic codes have theoretical importance, as they are connected to a number of areas of mathematics and outperform cyclic codes in several aspects. Negacyclic codes are a subclass of constacyclic codes and are distance-optimal in many cases. However, compared with the extensive study of cyclic codes, negacyclic codes are much less studied. In this paper, several families of ternary negacyclic codes and their duals are constructed and analysed. These families of negacyclic codes and their duals contain distance-optimal codes and have very good parameters in general.  

\vspace*{.3cm}
\end{abstract}

\section{Introduction and motivations} 

\subsection{Constacyclic codes}

For a given prime power $q$, let $\gf(q)$ denote the finite field with $q$ elements, and let $\gf(q)^*$ denote the multiplicative group of $\gf(q)$. A $q$-ary $[n, k, d]$ linear code $\C$ is a $k$-dimensional linear subspace of $\gf(q)^n$ with minimum distance $d$. Let $\lambda \in \gf(q)^*$. A $q$-ary linear code $\C$ of length $n$ is said to be $\lambda$-{\it constacyclic} if $(c_0,c_1,\ldots,c_{n-1})\in \C$ implies $(\lambda c_{n-1}, c_0,c_1,\ldots,c_{n-2})\in \C$. Let
\begin{align*}
\Phi: \ \gf(q)^n	&\rightarrow \gf(q)[x]/( x^n-\lambda)\\
(c_0,c_1, \ldots, c_{n-1})&\mapsto c_0+c_1x+c_2x^2+ \cdots + c_{n-1}x^{n-1}.
\end{align*}
It is  known each ideal of the quotient ring $\gf(q)[x]/(x^n-\lambda)$ is {\it principal} and a $q$-ary linear code $\C$ is $\lambda$-constacyclic if and only if $\Phi(\C)$ is an ideal of the quotient ring $\gf(q)[x]/(x^n-\lambda)$. Due to this fact, we will 
identify $\Phi(\C)$ with $\C$ for any $\lambda$-constacyclic code $\C$. 
Let $\C=( g(x))$ be a $q$-ary $\lambda$-constacyclic code, where $g(x)$ is monic and has the smallest degree. Then $g(x)$ is called the {\it generator polynomial} and $h(x)=(x^n-\lambda)/g(x)$ is referred to as the {\it check polynomial} of $\C$. A $q$-ary $\lambda$-constacyclic code $\C$ is said to be {\it irreducible} if its check polynomial is irreducible over $\gf(q)$. By definition, a $1$-constacyclic code is a {\it cyclic code}. In particular, $(-1)$-constacyclic code are called {\it negacyclic} codes. Hence, cyclic codes form a subclass of constacyclic codes. Further information on constacyclic codes can be found in  \cite{Black66,CDFL15,CFLL12,DP92,DDR11,DY10,FWF2017,KS90,LLLM17,LQL2017,LQ2018,MC21,PD91,WSZ19,WSD22,Wolfmann2008,SR2018,SF2020,SZW20,SWD22,ZSL18} and the references therein.

\subsection{Motivations and objectives}

Negacyclic codes over finite fields are a subclass of constacyclic codes, and were first studied by Berlekamp \cite{Black66} for correcting errors measured in the Lee metric. Therefore, the history of negacyclic codes goes back to 1966. In the past 56 years, some works on the application of negacyclic codes in quantum codes were done (see \cite{GLL20,KZ12,KZ13,KZ18,WLL20,ZSL18}). However, only a few references about theoretical results of negacyclic codes have appeared in the literature \cite{Black66,Black08,DDR11,PZS18,ZKZ19,WSD22},  Hence, very limited results on the parameters of negacyclic codes over finite fields are known in the literature.

Negacyclic codes have similar algebraic structures as cyclic codes. With the help of Magma, we found that 
the best ternary negacyclic code of certain length and dimension has a better error-correcting capability  than 
the best ternary cyclic code of the same length and dimension. Some examples of such code parameters are 
given in Table \ref{tab-sun1}.
\begin{table*}
	\begin{center}
	\renewcommand\arraystretch{1.1}
\caption{The best ternary cyclic codes and ternary negacyclic codes}\label{tab-sun1}
\begin{tabular}{ccc} \hline
  Best cyclic codes &  Best negacyclic codes & Best linear codes\\  \hline
  $[10,4,4]$ & $[10,4,6]$& $[10,4,6]$\\ \hline 
$[10,6,2]$ & $[10,6,4]$& $[10,6,4]$\\ \hline 
$[14,6,4]$ & $[14,6,6]$& $[14,6,6]$\\ \hline 
$[14,8,2]$ & $[14,8,5]$ & $[14,8,5]$\\ \hline 
$[20,10,6]$ & $[20,10,7]$ & $[20,10,7]$\\ \hline 
$[20,12,4]$ & $[20,12,5]$& $[20,12,6]$\\ \hline 
$[20,16,2]$ & $[20,16,3]$& $[20,16,3]$\\ \hline 
\end{tabular}
\end{center}
\end{table*}
Therefore, it is very interesting to study ternary negacyclic codes. This is the main motivation of studying 
ternary negacyclic codes in this paper. The objectives of this paper are the following: 
\begin{enumerate} 
\item Construct and analyse several families of ternary negacyclic codes. 
\item Study parameters of the duals of these ternary negacyclic codes.
\end{enumerate} 

 \subsection{The organisation of this paper}

The rest of this paper is organized as follows. In Section \ref{sec2}, we present some auxiliary results. 
In Section \ref{sec-negacodeevenlength}, we prove a general result for negacyclic codes of even length. 
In Section \ref{sec3}, we study the parameters of the first family of ternary negacyclic codes and their duals. 
In Section \ref{sec4}, we investigate the parameters of the second family of ternary negacyclic codes and their duals. In Section \ref{sec5}, we analyse the parameters of the third family of ternary negacyclic codes and their duals. In Section \ref{sec6}, we study the parameters of the fourth family of ternary negacyclic codes. These families of negacyclic codes and their duals have very good parameters. In Section \ref{sec7}, we conclude this paper and make some concluding remarks.

\section{Preliminaries}\label{sec2}

Throughout this section, let $q$ be an odd prime power, $\gf(q)$ be the finite field with $q$ elements and let $n$ be a positive integer with $\gcd(n, q)=1$. For a linear code $\C$, we use $\dim(\C)$ and $d(\C)$ to denote its dimension and minimum Hamming distance, respectively.

\subsection{Cyclotomic cosets} 
 To deal with $q$-ary negacyclic codes of length $n$, we need to define $q$-cyclotomic cosets modulo $2n$. 

Let $\Z_{2n}=\left\{0,1,2,\cdots,2n-1 \right\}$ be the ring of integers modulo $2n$. For any $i \in \Z_{2n}$, the \emph{$q$-cyclotomic coset of $i$ modulo $2n$} is defined by 
\[C^{(q,2n)}_i=\left\{i, iq, iq^2, \cdots, iq^{\ell_i-1}\right\} \bmod {2n} \subseteq \Z_{2n}, \]
where $\ell_i$ is the smallest positive integer such that $i \equiv i q^{\ell_i} \pmod{2n}$, and is the \textit{size} of the $q$-cyclotomic coset. The smallest integer in $C^{(q, 2n)}_i$ is called the \textit{coset leader} of $C^{(q, 2n)}_i$. Let $\Gamma_{(q,2n)}$ be the set of all the coset leaders. We have then $C^{(q, 2n)}_i \cap C^{(q, 2n)}_j = \emptyset$ for any two distinct elements $i$ and $j$ in  $\Gamma_{(q,2n)}$, and  
 $\bigcup_{i \in  \Gamma_{(q,2n)} } C_i^{(q,2n)}=\Z_{2n}$.

Let $m=\ord_{2 n}(q)$. Let $\alpha$ be a primitive element of $\gf(q^m)$ and let $\beta=\alpha^{(q^m-1)/2n}$. Then $\beta$ is a primitive $2n$-th root of unity in $\gf(q^m)$ and $\beta^n=-1$. The \textit{minimal polynomial} $\m_{\beta^i}(x)$ of $\beta^i$ over $\gf(q)$ is a monic polynomial of the smallest degree over $\gf(q)$ with $\beta^i$ as a zero. We have $
\m_{\beta^i}(x)=\prod_{j \in C_i^{(q,2n)}} (x-\beta^j) \in \gf(q)[x], 
$ 
which is irreducible over $\gf(q)$. It then follows that $ x^{2n}-1=\prod_{i \in  \Gamma_{(q,2n)}} \m_{\beta^i}(x)$. Define 
$$\Gamma_{(q,2n)}^{(1)}=\left\{i: i \in \Gamma_{(q,2n)}, \, i \equiv 1 ~({\rm mod}~2) \right\}.$$ Then $x^{n}+1=\prod_{i \in  \Gamma_{(q,2n)}^{(1)}} \m_{\beta^i}(x)$.

\subsection{Zeros, BCH bound and trace representation of negacyclic codes} 

Let $\C$ be a $q$-ary negacyclic code of length $n$ with generator polynomial $g(x)$. Then there is a subset $\Gamma \subseteq \Gamma_{(q, 2n)}^{(1)}$ such that $g(x)=\prod_{i\in \Gamma}\m_{\beta^i}(x)$. Let $T=\cup_{i\in \Gamma}C_i^{(q,2n)}$. The roots of unity $\mathcal{Z}(\C)=\left\{\beta^i:\ i\in T \right\}$ are called the {\it zeros} of the negacyclic code $\C$ and $$\left\{\beta^i:\ i\in \Z_{2n}\backslash T~{\rm and~}~i~{\rm is~ odd} \right\}$$ are the {\it nonzeros} of $\C$. It is easily checked that $\dim(\C)=n-|\mathcal{Z}(\C)|$. The minimum distance of the negacyclic code $\C$ has the following lower bound.

\begin{lemma}\label{lem1}\cite[Lemma 4]{KS90}[The BCH bound for negacyclic codes] 
Let $\mathcal{C}$ be a $q$-ary negacyclic code of length $n$ with zeros $\mathcal{Z}(\C)$. If there are integers $h$ and $\delta$ with $2\leq \delta \leq n$ such that $$\left\{\beta^{1+2i}: h\leq i\leq h+\delta -2 \right\}\subseteq \mathcal{Z}(\C).$$ Then $d(\C)\geq \delta$.
\end{lemma}  

The trace representation of negacyclic codes is documented below (see \cite{DY10}, \cite{SR2018}, \cite[Theorem 1]{SZW20}). 

\begin{lemma}\label{lem-003} 
Let $n$ be a positive integer such that $\gcd(n, q)=1$. Define $m=\ord_{2n}(q)$ and let $\beta \in  \gf(q^m)$ be a primitive $2n$-th root of unity. Let $\C$ be the $q$-ary negacyclic code of length $n$ with check polynomial $\prod_{j=1}^s \m_{\beta^{i_j}}(x)$, where $C_{i_{a}}^{(q, 2n)} \cap C_{i_{b}}^{(q, 2 n)}=\emptyset$ for $a\neq b$. Then $\C$ has the trace representation 
$$\left \{ \left(\sum_{j=1}^s{\rm Tr}_{q^{m_j}/q}(a_j\beta^{-ti_j})\right)_{t=0}^{n-1} \,:\,a_j\in {\rm GF}(q^{m_j}),\,1\leq j\leq s \right \},$$ 
where $m_j=|C_{i_j}^{(q, 2n)}|$ and $\tr_{q^m/q}$ denotes the trace function from $\gf(q^m)$ to $\gf(q)$.
\end{lemma}

\subsection{The duals of negacyclic codes}

Let $\C$ be a $q$-ary linear code of length $n$. Then its {\it dual code}, denoted by $\C^{\bot}$, is defined by
$$\C^{\bot}=\left\{(c_0,c_1,\ldots,c_{n-1}) \in \gf(q)^n:\  \sum_{i=0}^{n-1}c_ib_i=0, \ \forall \ (b_0,b_1,\ldots,b_{n-1})\in \C \right\}.$$
Similar to classical cyclic codes, the dual codes of negacyclic codes are also negacyclic codes. Let $f(x)=f_0+f_1x+\cdots+f_t x^t\in \gf(q)[x]$, where $f_0 f_t\neq 0$ and $t$ is a positive integer. The {\it reciprocal polynomial} of $f(x)$, denoted by $\widehat{f}(x)$, is defined by $\widehat{f}(x)=f_0^{-1}x^t f(x^{-1})$. Negacyclic codes and their duals have the following relation, which is a fundamental result.  

\begin{lemma}\cite{KS90}
Let $\C$ be a $q$-ary negacyclic code of length $n$ generated by $g(x)$. Then the dual code of $\C$ is the $q$-ary negacyclic code of length $n$ generated by $\widehat{h}(x)$, where $h(x)=(x^n+1)/g(x)$.
\end{lemma}

\subsection{Bounds of linear codes }

We recall the following two bounds on linear codes, which will be needed in the sequel.

\begin{lemma}\label{lem2}
{\rm (Sphere Packing Bound \cite{HP2003})} Let $\C$ be a $q$-ary $[n, k, d]$ code, then 
$$\sum_{i=0}^{\lfloor \frac{d-1}2\rfloor} \binom{n}{i} (q-1)^i\leq q^{n-k},$$
where $\lfloor \cdot \rfloor$ is the floor function.
\end{lemma}

The following lemma is the sphere packing bound for even minimum distances.

\begin{lemma}{\rm \cite{FWF2017}}\label{lem3}
Let $\C$ be a $q$-ary $[n,k,d]$ code, where $d$ is an even integer.	Then
$$\sum_{i=0}^{\frac{d-2}2} \binom{n-1}{i} (q-1)^i\leq q^{n-1-k}.$$
\end{lemma}

A $q$-ary $[n, k, d]$ code is said to be {\it distance-optimal} if there is no $q$-ary $[n, k, d']$ code with $d'>d$. A $q$-ary $[n, k, d]$ code is said to be {\it dimension-optimal} if there is no $q$-ary $[n, k', d]$ code with $k'> k$. A $q$-ary $[n, k, d]$ code is said to be {\it length-optimal} if there is no $q$-ary $[n', k, d]$ code with $n'< n$. A linear code is said to be {\it optimal} if it is distance-optimal, or dimension-optimal or length-optimal. 


\subsection{Several equivalences of linear codes}

In this paper, we will need the following notions: 
\begin{itemize}
\item The permutation automorphism group of a linear code $\C$ denoted by $\PAut(\C)$ and 
          the permutation equivalence of two linear codes over a finite field. 
\item The monomial automorphism group of a linear code $\C$ denoted by $\MAut(\C)$ and 
          the monomial equivalence of two linear codes over a finite field. 
\item The automorphism group of a linear code $\C$ denoted by $\Aut(\C)$ and 
          the equivalence of two linear codes over a finite field.           
\end{itemize} 
For definitions of these concepts, the reader is referred to \cite[Section 2.8]{dingtang2022} or \cite[Chapter 1]{HP2003}. 
Two $q$-ary linear codes $\C_1$ and $\C_2$ are said to be \emph{scalar-equivalent\index{scalar-equivalent}} if there is an invertible diagonal matrix $D$ over $\gf(q)$ such that $\C_2=\C_1D$. 

When $n$ is an odd positive integer,  let 
\begin{align*}
\psi: \ \gf(q)[x]/(x^n+1)\rightarrow \gf(q)[x]/( x^n-1), \ c(x)\mapsto c(-x).
\end{align*}
It is easily checked that $\psi$ is a ring isomorphism. Define $\psi(\C)=\{\psi(\bc):\ \bc \in \C  \}$. Then we have the following result.

\begin{theorem}\cite{Black08}
Let $n$ be an odd positive integer. Then $\C$ is a $q$-ary negacyclic code of length $n$ if and only if $\psi(\C)$ is a $q$-ary cyclic code of length $n$. Furthermore, the negacyclic code $\C$ is scalar-equivalent to the cyclic code $\psi(\C)$.
\end{theorem}

Table \ref{tab-sun1} shows that there are distance-optimal negacyclic codes that are not scalar-equivalent to cyclic codes.
Although a negacyclic code $\C$ of odd length $n$ is scalar-equivalent to a cyclic code $\psi(\C)$, it is still valuable 
to study $\C$, as the parameters of the cyclic code $\psi(\C)$ may still be open in the literature.

\section{A theorem about negacyclic codes of even length}\label{sec-negacodeevenlength} 

Let $n$ be an even integer and $q\equiv 1\pmod{4}$. Then there is a $\lambda\in \gf(q)^*$ such that $\lambda^2=-1$. It follows that $x^n+1=(x^{\frac{n}2}-\lambda)(x^{\frac{n}2}+\lambda)$. Let $e_1(x)=\frac{\lambda}2(x^{\frac{n}2}-\lambda)$ and $e_2(x)=-\frac{\lambda}2(x^{\frac{n}2}+\lambda)$. In the ring $R:=\gf(q)[x]/(x^n+1)$, it is easily checked that $e_1(x)+e_2(x)=1$, $e_1(x)^2=e_1(x)$, $e_2(x)^2=e_2(x)$ and $e_1(x)e_2(x)=0$. Then $R$ is the direct sum of the ideals generated by the $e_i(x)$, i.e., $$R=e_1(x)R+e_2(x)R.$$ Furthermore, $e_1(x)R\cong \gf(q)[x]/(x^{\frac{n}2}+\lambda) $ and $e_2(x)R\cong \gf(q)[x]/(x^{\frac{n}2}-\lambda) $. Let $\C$ be a $q$-ary negacyclic code of length $n$. We associate to $\C$ the following two $q$-ary linear codes:
\begin{align*}
{\rm Res}_1(\C)&=\left\{c(x)~({\rm mod}~x^{\frac{n}2}+\lambda):\ c(x)\in \C \right\},\\
{\rm Res}_2(\C)&=\left \{c(x)~({\rm mod}~x^{\frac{n}2}-\lambda): \ c(x)\in \C \right \}.	
\end{align*}
By definition, the code ${\rm Res}_1(\C)$ (resp. ${\rm Res}_2(\C)$) is a $q$-ary $(-\lambda)$-constacyclic (resp. $\lambda$-constacyclic) code of length $\frac{n}2$. Moreover, we have the following results.

\begin{theorem}\label{THM05}
Let $n$ be an even integer and $q\equiv 1~({\rm mod}~4)$.  Let $\C$ be a $q$-ary negacyclic code of length $n$ with generator polynomial $g(x)$. Then 
\begin{eqnarray*}
\C &=& e_1{\rm Res}_1(\C)+e_2{\rm Res}_2(\C) \\ 
&=& \left\{e_1(x)c_1(x)+e_2(x)c_2(x): \ c_1(x)\in {\rm Res}_1(\C), c_2(x)\in {\rm Res}_2(\C) \right\}.
\end{eqnarray*} 
Furthermore, the following hold.
\begin{enumerate}
\item The code ${\rm Res}_1(\C)$ is the $q$-ary $(-\lambda)$-constacyclic code of length $\frac{n}2$ with generator polynomial $g_1(x)=\gcd(g(x), x^{\frac{n}2}+\lambda)$.
\item The code ${\rm Res}_2(\C)$ is the $q$-ary $\lambda$-constacyclic code of length $\frac{n}2$ with generator polynomial 
$g_2(x)=\gcd(g(x), x^{\frac{n}2}-\lambda).$

\item $\dim(\C)=\dim({\rm Res}_1(\C))+\dim({\rm Res}_2(\C))$.
\item If $(x^{\frac{n}2}-\lambda) \mid g(x)$, then $d(\C)= 2\cdot d({\rm Res}_1(\C))$.
\item If $(x^{\frac{n}2}+\lambda) \mid g(x)$, then $d(\C)= 2\cdot d({\rm Res}_2(\C))$.
\item If $(x^{\frac{n}2}-\lambda) \nmid g(x)$ and $(x^{\frac{n}2}+\lambda) \nmid g(x)$, then 
$$ d(\C)= 2\cdot \min \{d({\rm Res}_1(\C)),d({\rm Res}_2(\C)) \},$$ provided that $$2\cdot \min \{d({\rm Res}_1(\C)),d({\rm Res}_2(\C)) \} \leq \max\{d({\rm Res}_1(\C)),d({\rm Res}_2(\C)) \},$$
and 
$$\max\{d({\rm Res}_1(\C)),d({\rm Res}_2(\C)) \}\leq   d(\C)\leq 2\cdot \min \{d({\rm Res}_1(\C)),d({\rm Res}_2(\C)) \},$$   
provided that $$2\cdot \min \{d({\rm Res}_1(\C)),d({\rm Res}_2(\C)) \}>\max\{d({\rm Res}_1(\C)),d({\rm Res}_2(\C)) \}.$$
\end{enumerate}	
\end{theorem}

\begin{proof}
Firstly, we prove that ${\rm Res}_i(\C)$ is the $q$-ary constacyclic code generated by $g_i(x)$. Let $\mathcal{D}$ be the $q$-ary $(-\lambda)$-constacyclic code of length $\frac{n}2$ with generator polynomial $g_1(x)$. 
Let $c_1(x)\in {\rm Res}_1(\C)$, then there is $c(x)\in \C$ such that $c_1(x)\equiv c(x)~({\rm mod}~x^{\frac{n}2}+\lambda)$. It is clear that $$\gcd(c(x),x^{\frac{n}2}+\lambda)=\gcd(c_1(x),x^{\frac{n}2}+\lambda).$$
 Then $g_1(x)$ divides $c_1(x)$. It follows that ${\rm Res}_1(\C)\subseteq \mathcal{D}$. On the other hand, let $c_2(x)\in \mathcal{D}$, then $g_1(x)$ divides $c_2(x)$. Note that $\gcd\left(\frac{x^{\frac{n}2}+\lambda}{g_1(x)},\frac{g(x)}{g_1(x)} \right)=1$, then there are $a_1(x)$ and $a_2(x)$ such that $$a_1(x)\left(\frac{x^{\frac{n}2}+\lambda}{g_1(x)}\right)+a_2(x)\left(\frac{g(x)}{g_1(x)}\right)=1.$$ It follows that 
$$c(x):=a_2(x)\left(\frac{g(x)}{g_1(x)}\right) c_2(x)=c_2(x)-a_1(x)\left(\frac{x^{\frac{n}2}+\lambda}{g_1(x)}\right)c_2(x).$$
Since $g_1(x)\mid c_2(x)$, we have $g(x)\mid c(x)$ and $c_2(x)\equiv c(x)~({\rm mod}~x^{\frac{n}2}+\lambda) $. Therefore, $c_2(x)\in {\rm Res}_1(\C)$. It follows that $\mathcal{D}\subseteq {\rm Res}_1(\C)$. Consequently, ${\rm Res}_1(\C)=\mathcal{D}$, i.e., ${\rm Res}_1(\C)$ is the $q$-ary $(-\lambda)$-constacyclic code of length $\frac{n}2$ with generator polynomial $g_1(x)$. By a similar method, we can prove that ${\rm Res}_2(\C)$ is the $q$-ary $\lambda$-constacyclic code of length $\frac{n}2$ with generator polynomial $g_2(x)$.

Secondly, we prove that $\C=e_1{\rm Res}_1(\C)+e_2{\rm Res}_2(\C)$. Since $1=e_1(x)+e_2(x)$, any $c(x)\in \C$ can be written uniquely in the form 
\begin{align*}
c(x)&=e_1(x)c(x)+e_2(x)c(x)\\
&=e_1(x) c_1(x)+e_2(x) c_2(x),	
\end{align*}
where $c_1(x)\equiv c(x)~({\rm mod}~x^{\frac{n}2}+\lambda)\in {\rm Res}_1(\C)$ and $c_2(x)\equiv c(x)~({\rm mod}~x^{\frac{n}2}-\lambda)\in {\rm Res}_2(\C)$. Therefore, $\C \subseteq e_1{\rm Res}_1(\C)+e_2{\rm Res}_2(\C)$. On the other hand, $\dim({\rm Res}_1(\C))=\frac{n}2-\deg(g_1(x))$ and $\dim({\rm Res}_2(\C))=\frac{n}2-\deg(g_2(x))$. It is easily checked that $g(x)=g_1(x)g_2(x)$. Then 
\begin{align*}
\dim(\C)&=n-\deg(g(x))\\
&=n-\deg(g_1(x))-\deg(g_2(x))\\
&=\dim({\rm Res}_1(\C))+\dim({\rm Res}_2(\C)).	
\end{align*}
Therefore, $|\C|=|{\rm Res}_1(\C)|\cdot |{\rm Res}_2(\C)|=|e_1{\rm Res}_1(\C)+e_2{\rm Res}_2(\C)|$. Consequently, $\C=e_1{\rm Res}_1(\C)+e_2{\rm Res}_2(\C)$. 

Finally, we study the minimum distance of $\C$. Let $c(x)=e_1(x)c_1(x)+e_2(x)c_2(x)\in \C$ and $c(x)\neq 0$, where $c_1(x)\in {\rm Res}_1(\C)$ and $c_2(x)\in {\rm Res}_2(\C)$ are not all $0$. Note that
\begin{align*}
c(x)=\frac{\lambda}2(c_1(x)-c_2(x)) x^{\frac{n}2}+\frac{1}2(c_1(x)+c_2(x)),	
\end{align*}
we have
\begin{align}\label{EEE-1}
\wt(c(x))=\wt(c_1(x)-c_2(x))+\wt(c_1(x)+c_2(x)).
\end{align}
There are the following three cases.
\begin{itemize}
\item If $(x^{\frac{n}2}-\lambda) \mid g(x)$, then ${\rm Res}_2(\C)=\{\0\}$, i.e., $c_2(x)\equiv 0$. It follows from (\ref{EEE-1}) that $$\wt(c(x))=2\cdot \wt(c_1(x)).$$ Therefore, $d(\C)=2\cdot d({\rm Res}_1(\C))$. 	
\item If $(x^{\frac{n}2}+\lambda) \mid g(x)$, then ${\rm Res}_1(\C)=\{\0\}$. i.e., $c_1(x)\equiv 0$. It follows from (\ref{EEE-1}) that $$\wt(c(x))=2\cdot \wt(c_2(x)).$$ Therefore, $d(\C)=2\cdot d({\rm Res}_2(\C))$. 
\item If $(x^{\frac{n}2}-\lambda) \nmid g(x)$ and $(x^{\frac{n}2}+\lambda) \nmid g(x)$, then ${\rm Res}_i(\C)\neq \{\0\}$ for $i\in \{1,2\}$. On one hand, if $c_1(x)=0$ (resp. $c_2(x)=0$), from (\ref{EEE-1}), $\wt(c(x))=2\cdot \wt(c_2(x))$ (resp. $\wt(c(x))=2\cdot \wt(c_1(x))$). For any two elements $a$ and $b$ in $\gf(q)$, we have $a-b=0$ and $a+b=0$ if and only if $a=b=0$, as $q$ is odd.
On the other hand, if $c_1(x)\neq 0$ and $c_2(x)\neq 0$, then 
\begin{align*}
\wt(c(x))\geq \max \{ \wt(c_1(x)), \wt(c_2(x))\}\geq  \max \{ d({\rm Res}_1(\C)), d({\rm Res}_2(\C))\}.
\end{align*}
Combining the results above with the discussions of the two cases above, we have the following conclusions.
\begin{itemize}
\item If $2 \cdot \min \{d({\rm Res}_1(\C)),d({\rm Res}_2(\C))\}\leq \max \{ d({\rm Res}_1(\C)), d({\rm Res}_2(\C))\}$, we have $$d(\C)=2 \cdot \min \{d({\rm Res}_1(\C)),d({\rm Res}_2(\C))\}.$$
\item If $\max \{ d({\rm Res}_1(\C)), d({\rm Res}_2(\C))\}<2 \cdot \min \{d({\rm Res}_1(\C)),d({\rm Res}_2(\C))\}$, we have $$\max \{ d({\rm Res}_1(\C)), d({\rm Res}_2(\C))\}\leq d(\C)\leq 2 \cdot \min \{d({\rm Res}_1(\C)),d({\rm Res}_2(\C))\}.$$
\end{itemize}
\end{itemize}
This completes the proof.
\end{proof}

For two $q$-ary linear codes $\C_1$, $\C_2$ of length $n$, the {\it $u+v|u-v$ construction} (see \cite{Hughes20}) is defined by
$$\C_1 \curlyvee \C_2=\{(\bc_1+\bc_2, \bc_1-\bc_2):\ \bc_1\in \C_1, \bc_2\in \C_2 \}.$$
Theorem \ref{THM05} shows that if $n$ is even and $q\equiv 1\pmod{4}$, then the $q$-ary negacyclic code $\C$ of length $n$ generated by $g(x)$ is scalar-equivalent to ${\rm Res}_1(\C) \curlyvee {\rm Res}_2(\C)$. 

Theorem \ref{THM05} can be used to study a $q$-ary negacyclic code $\C$ of even length $n$ via the study of the two associated 
constacyclic codes ${\rm Res}_1(\C)$ and ${\rm Res}_2(\C)$ of length $\frac{n}{2}$ if $q \equiv 1 \pmod{4}$. Hence, it may be very useful in some 
cases. However, it cannot be used to study ternary negacyclic codes of even length.

\section{The first family of ternary negacyclic codes and their duals}\label{sec3}

In this section, let $\rho>3$ be an odd prime, we will construct a family of ternary irreducible negacyclic codes of length $2\rho$. To settle the value of $\ord_{4\rho}(3)$, we need the following two lemmas.

\begin{lemma}\cite{IR1990}\label{lem:4}
$3$ is a quadratic residue of primes of the form $12\ell+1$ and $12\ell-1$. $3$ is a quadratic nonresidue of primes of the form $12\ell+5$ and $12\ell-5$. 
\end{lemma}

\begin{lemma} \label{LLEM-1}
Let $\rho>3$ be an odd prime, then $\ord_{4\rho}(3)=\rho-1$ if and only if the one of the following conditions holds:
\begin{enumerate}
	\item $\rho \equiv \pm 5 ~({\rm mod}~12)$ and $\ord_{\rho}(3)=\rho-1$.
	\item $\rho \equiv -1  ~({\rm mod}~12)$ and $\ord_{\rho}(3)=\frac{\rho-1}2$.
\end{enumerate}	
\end{lemma}

\begin{proof}
($\Rightarrow$) Suppose $\ord_{4\rho}(3)=\rho-1$. Let 
$\ell=\ord_{\rho}(3)$. It is clear that $\ord_{\rho}(3)$ divides $\ord_{4\rho}(3)$, i.e., $\ell\mid(\rho-1)$. 

\begin{enumerate}
\item Suppose that $\ell$ is even. Then $4\mid(3^{\ell}-1)$. Note that $\rho\mid(3^{\ell}-1)$ and $\gcd(4,\rho)=1$, we have $4\rho \mid(3^{\ell}-1)$. It follows that $\ord_{4\rho}(3)=\rho-1$ divides $\ell$. Consequently, $\ell=\rho-1$. It then follows that $3$ is a quadratic nonresidue of $\rho$. By Lemma \ref{lem:4}, $\rho \equiv \pm 5~({\rm mod}~12)$.	
\item Suppose that $\ell$ is odd. It follows from $\ell\mid(\rho-1)$ that $\ell \mid(\frac{q-1}2)$. Consequently, $3$ is a quadratic residue of $\rho$. By Lemma \ref{lem:4}, $\rho \equiv \pm 1 ~({\rm mod}~12)$. We consider the following two cases:
\begin{enumerate}
\item $\rho \equiv 1~({\rm mod}~12)$, then $\frac{\rho-1}2$ is even. It follows that $3^{\frac{\rho-1}2}\equiv 1~({\rm mod}~4\rho)$, which contradicts the fact that $\ord_{4\rho}(3)=\rho-1$. 	
\item $\rho \equiv -1~({\rm mod}~12)$. Note that $4\rho\mid (3^{2\ell}-1)$, we have $\ord_{4\rho}(3)$ divides $2\ell$, i.e., $(\frac{\rho-1}2)\mid \ell$. Therefore, $\ell=\frac{\rho-1}2$.  
\end{enumerate}
\end{enumerate}

($\Leftarrow$) Let $\ell'=\ord_{4\rho}(3)$. Notice that $\phi(4\rho)=2(\rho-1)$, where $\phi$ is the Euler totient function. Then  $\ell'$ divides $2(\rho-1)$. According to the primitive root theorem, $\ell'\neq 2(\rho-1)$. Therefore, $\ell'\leq \rho-1$. It is clear that $\ord_{\rho}(3)\mid \ell'$.   We consider the following two cases:
\begin{enumerate}
\item Suppose that $\ord_{\rho}(3)=\rho-1$. Then $\ell'=\rho-1$.
\item Suppose that $\rho\equiv -1\pmod{12}$ and $\ord_{\rho}(3)=\frac{\rho-1}2$. Note that $\ord_{3}(\rho)$ divides $\ell'$ and $\ell'\leq \rho-1$, then $\ell'=\rho-1$ or $\frac{\rho-1}2$. In this case, $\frac{\rho-1}2$ is odd, $\gcd(3^{\frac{\rho-1}2}-1,4)=2$. It follows that $\ell'\neq \frac{\rho-1}2$. Therefore, $\ell'=\rho-1$. 
\end{enumerate}
This completes the proof.
\end{proof}

Throughout the rest of this section, let $\rho$ be an odd prime such that $\rho \equiv \pm 5 ~({\rm mod}~12)$ and $\ord_{\rho}(3)=\rho-1$ 
and let $n=2\rho$. By Lemma \ref{LLEM-1}, $\ord_{4\rho}(3)=\rho-1$. By assumption, 
$$
C_1^{(3, 4\rho)}=\{1, 3, 3^2, \cdots, 3^{\rho-2}\} \bmod{4\rho}. 
$$  
It is clear that $C_\rho^{(3, 4\rho)}=\{\rho, 3\rho\}$ and all elements in both $C_1^{(3, 4\rho)}$ and $C_\rho^{(3, 4\rho)}$ 
are odd. Let $h$ be any integer from 
$$
\{2i+1: 0 \leq i \leq 2\rho-1\} \setminus (C_1^{(3, 4\rho)} \cup C_\rho^{(3, 4\rho)}).  
$$
Note that $\gcd(h,4\rho)=1$, we have $|C_h^{(3,4\rho)}|=|C_1^{(3,4\rho)}|=\rho-1$. Then $C_1^{(3, 4\rho)}$, $C_h^{(3, 4\rho)}$ and $C_\rho^{(3, 4\rho)}$ form a partition of $\{2i+1: 0 \leq i \leq 2\rho-1\}$. 

Let $\alpha$ be a primitive element of $\gf(3^m)$, where $m=\rho-1$. Put 
$$
\beta= \alpha^{(3^{\rho-1}-1)/4\rho}. 
$$ 
By definition, $\beta^{2\rho}=-1$. 

Let $\m_{\beta^i}(x)$ denote the minimal polynomial of $\beta^i$ over $\gf(3)$. It is easily seen that 
$$
\m_{\beta^\rho}(x)=x^2+1
$$ 
and 
\begin{eqnarray*}
x^n+1 = \m_{\beta^\rho}(x) \m_{\beta}(x) \m_{\beta^h}(x). 
\end{eqnarray*}

Let $\C(\rho)$ denote the ternary negacyclic code of length $2\rho$ with check polynomial $\m_{\beta}(x)$. Then $\dim(\C(\rho))=\rho-1$. By definition, $\C(\rho)$ is irreducible. By Lemma \ref{lem-003}, the code $\C(\rho)$ has the trace representation 
\begin{equation}\label{EQ1}
\C(\rho)=\left\{\bc(a)=\left(\tr_{{3^m}/3}(a \theta^i) \right)_{i=0}^{2\rho-1}:\ a\in \gf(3^m)\right\},
\end{equation}
where $\theta=\beta^{-1}$. The dual code $\C(\rho)^{\bot}$ has the trace representation
\begin{equation*}
\C(\rho)^{\bot}=\left\{ \bc(a, b)=\left(\tr_{{3^2}/3}(a \beta^{\rho i} )+\tr_{{3^m}/3}(b \beta^{h i} ) \right)_{i=0}^{2\rho-1}:\ a \in \gf(3^2) , \ b\in \gf(3^m)\right\}.
\end{equation*}
 Associated with the irreducible negacyclic code $\C(\rho)$ is the following code over $\gf(3^2)$:
\begin{equation*}
\overline{\C}(\rho)=\left\{\overline{\bc}(a)=\left(\tr_{{3^m}/{3^2}}(a \theta^{2i})\right)_{i=0}^{\rho-1}: \ a \in \gf(3^m)  \right \}.	
\end{equation*}
Note that $\theta^2$ is a primitive $2\rho$-th root of unity and $\theta^{2\rho}=-1$. It follows from Lemma \ref{lem-003} that $\overline{\C}(\rho)$ is the $3^2$-ary negacyclic code of length $\rho$ with check polynomial $\m_{\beta^2}(x)$. It is easily verified that the dual code $\overline{\C}(\rho)^{\bot}$ has the trace representation
\begin{equation*}
\overline{\C}(\rho)^{\bot}=\left\{ \overline{\bc}(a, b)=\left(a(-1)^i+\tr_{{3^m}/{3^2}}(b \beta^{2hi})\right)_{i=0}^{\rho-1}:\ a \in \gf(3^2) , \ b\in \gf(3^m)\right\}.
\end{equation*}
We first prove the following theorem.

\begin{theorem}
	Let notation be the same as before. The $3^2$-ary negacyclic code $\overline{\C}(\rho)$ has parameters $\left[ \rho, \frac{\rho-1}2, d\geq \sqrt{\rho}+1 \right]$, and $\overline{\C}(\rho)^{\bot}$ has parameters $\left[ \rho, \frac{\rho+1}2, d\geq \sqrt{\rho} \right]$.
\end{theorem}

\begin{proof}
Suppose $\ord_{2\rho}(3^2)=\ell$. Since $4\rho \mid(3^{\rho-1}-1)$, we have $2\rho \mid(3^{\rho-1}-1)$. It follows that $\ell \mid (\frac{\rho-1}2)$. On the other hand, since $2\rho \mid(3^{2\ell}-1)$ and $4\mid (3^{2\ell}-1)$, we have $4\rho \mid(3^{2\ell}-1)$. It then follows that $\ord_{4\rho}(3)\mid 2\ell$, i.e., $(\frac{\rho-1}2) \mid \ell$. Therefore, $\ell=\frac{\rho-1}2$. According to \cite{SWD22}[Theorem 26], the irreducible negacyclic code $\overline{\C}(\rho)$ has parameters $[\rho, \frac{\rho-1}2, d\geq \sqrt{\rho}+1]$, and $\overline{\C}(\rho)^{\bot}$ has parameters $[\rho, \frac{\rho+1}2, d\geq \sqrt{\rho}]$. This completes the proof. 
\end{proof}

\begin{theorem}\label{thm-am1}
Let notation be the same as before. Then the ternary negacyclic code $\C(\rho)$ has parameters $\left[2\rho, \rho-1, d\right]$, where $d\geq d(\overline{\C}(\rho))\geq \sqrt{\rho}+1$.
\end{theorem}

\begin{proof} 
It is clear that $\dim(\C(\rho))=\ord_{4\rho}(3)=\rho-1$. Now we prove that $d(\C(\rho))\geq d(\overline{\C}(\rho))$. 
Since $\gcd(2,\rho)=1$, we have $\left\{2i+\rho j: \ 0\leq i\leq \rho-1, \ 0\leq j\leq 1 \right \}$ is a complete set of residues modulo $2\rho$. It then follows that
\begin{eqnarray*}
\left\{0,1,\cdots, 2\rho-1 \right\} &=& \left\{2i:~0\leq i\leq \rho-1 \right\} \cup \left\{2i+\rho:~0\leq i\leq \frac{\rho-1}2 \right\} \\
&& \cup  \left\{2i+\rho-2\rho:~\frac{\rho+1}2\leq i\leq \rho-1 \right\}.
\end{eqnarray*}
For any codeword $ \bc(a)= \left(\tr_{{3^m}/3}(a \theta^i ) \right)_{i=0}^{2\rho-1} \in \C(\rho)$, it is easily checked that $\bc(a)$ is permutation-equivalent to
\begin{align*}
& \left( \left(\tr_{{3^m}/3}(a \theta^{2i} ), \tr_{{3^m}/3}(a \theta^{2i+\rho} )\right)_{i=0}^{\frac{\rho-1}2}\| \left(\tr_{{3^m}/3}(a \theta^{2i} ), \tr_{{3^m}/3}(a \theta^{2i+\rho -2\rho} )\right)_{i=\frac{\rho+1}2}^{\rho-1}\right) \\
 =&\left( \left(\tr_{{3^m}/3}(a \theta^{2i} ), \tr_{{3^m}/3}(a \theta^{2i+\rho} )\right)_{i=0}^{\frac{\rho-1}2}\| \left(\tr_{{3^m}/3}(a \theta^{2i} ), -\tr_{{3^m}/3}(a \theta^{2i+\rho} )\right)_{i=\frac{\rho+1}2}^{\rho-1}\right),
\end{align*}
where $\|$ denotes the concatenation of vectors. It then follows that the code $\C(\rho)$ is monomial-equivalent to 
\begin{equation}\label{EQ5}
\widetilde{\C}(\rho)=\left\{\widetilde{\bc}(a)=\left((\tr_{3^m/3}(a \theta^{2i}),\tr_{3^m/3}(a  \theta^{2i+\rho}) ) \right)_{i=0}^{\rho-1}:\ a\in \gf(3^m)\right\}.
\end{equation}
For any $a\in \gf(3^m)^*$, it follows from (\ref{EQ5}) that 
\begin{align} \label{EQ6}
\wt(\widetilde{\bc}(a))&\geq \rho-\left| \left\{ 0\leq i\leq \rho-1:\ \tr_{3^m/3}(a \theta^{2i})=\tr_{3^m/3}(a \theta^{2i+\rho})=0\right\}\right|	 \notag \\ \notag
&=\rho-\frac{1}{9}\sum_{i=0}^{\rho-1} \sum_{x_1\in \gf(3)} \zeta_3^{x_1\tr_{3^m/3}(a \theta^{2i})}  \sum_{x_2\in \gf(3)} \zeta_3^{x_2\tr_{3^m/3}(a \theta^{2i+\rho})} \\
&=\rho-\frac{1}{9}\sum_{i=0}^{\rho-1} \sum_{(x_1,x_2)\in \gf(3)^2} \zeta_3^{\tr_{3^m/3}( (x_1+x_2\theta^{\rho}) a\theta^{2i})},  
\end{align}
where $\zeta_3=e^{2\pi \sqrt{-1}/3}$. Note that $\ord(\theta^{\rho})=4$, we have $\theta^{\rho}\in \gf(3^2)\backslash \gf(3) $. It is easy to verify that 
\begin{equation}\label{EQ7}
	\left\{x_1+x_2\theta^{\rho}: \ x_1, \ x_2\in \gf(3) \right \}=\gf(3^2).
\end{equation}
It follows from (\ref{EQ6}) and (\ref{EQ7}) that
\begin{align*}
\wt(\widetilde{\bc}(a))&\geq \rho-\frac{1}{9}\sum_{i=0}^{\rho-1} \sum_{ x\in \gf(3^2)} \zeta_3^{\tr_{3^m/3}(x a\theta^{2i})}\notag \\
&=	 \rho-\frac{1}{9}\sum_{i=0}^{\rho-1} \sum_{ x\in \gf(3^2)} \zeta_3^{\tr_{3^2/3}( x\tr_{3^m/{3^2}} (a\theta^{2i}))}\notag \\
&=\left|\left\{ 0\leq i\leq \rho-1:\ \tr_{3^m/{3^2}}(a \theta^{2i})\neq 0\right\}\right| \notag \\
&=\wt(\overline{\bc}(a) )\geq d(\overline{\C}(\rho)).
\end{align*}
Therefore, $d(\C(\rho))\geq d(\overline{\C}(\rho))$. The desired result follows.
\end{proof}

\begin{theorem}\label{thm-am2}
Let notation be the same as before. Then the ternary negacyclic code $\C(\rho)^{\bot}$ has parameters $\left[2\rho, \rho+1, d^{\bot}\right]$, where $d^{\bot}\geq d(\overline{\C}(\rho)^{\bot})\geq \sqrt{\rho}$.
\end{theorem}

\begin{proof}
The dimension of the code $\C(\rho)^{\bot}$ follows from $\dim(\C(\rho))=\rho-1$. Now we prove that $d(\C(\rho)^{\bot})\geq d(\overline{\C}(\rho)^{\bot})$. It is similarly verified that the code $\C(\rho)^{\bot}$ is monomial-equivalent to 
\begin{align}\label{EQ9}
\left\{\widetilde{\bc}(a, b)=\left( \left(\tr_{{3^2}/3}(a \beta^{\rho 2i} )+\tr_{{3^m}/3}(b \beta^{h2i} ), \tr_{{3^2}/3}(a \beta^{\rho (2i+\rho)} )+\tr_{{3^m}/3}(b \beta^{h(2i+\rho)} ) \right) \right)_{i=0}^{\rho-1}: \right.\notag\\
\left. a\in \gf(3^2), \ b \in \gf(3^m)\right\}.
\end{align}
For any $(a, b) \in (\gf(3^2)\times \gf(3^m))\backslash (0,0)$, it follows from (\ref{EQ9}) that 
\begin{align} \label{EQ10}
\wt(\widetilde{\bc}(a,b))&\geq \rho-\left| \left \{ 0\leq i\leq \rho-1:\ \tr_{{3^2}/3}(a \beta^{\rho 2i} )+\tr_{{3^m}/3}(b \beta^{h2i} ) =0, \right. \right. \notag \\
&\ \ \ \ \  \  \ \ \ \ \ \ \ \left. \left. \tr_{{3^2}/3}(a \beta^{\rho (2i+\rho)} )+\tr_{{3^m}/3}(b \beta^{h(2i+\rho)} )=0 \right \} \right|	 \notag \\ \notag
&=\rho-\frac{1}{9}\sum_{i=0}^{\rho-1} \sum_{x_1\in \gf(3)} \zeta_3^{x_1(\tr_{{3^2}/3}(a \beta^{\rho 2i} )+\tr_{{3^m}/3}(b \beta^{h2i} )) }  \sum_{x_2\in \gf(3)} \zeta_3^{x_2(\tr_{{3^2}/3}(a \beta^{\rho (2i+\rho)} )+\tr_{{3^m}/3}(b \beta^{h(2i+\rho)} ))} \\
&=\rho-\frac{1}{9}\sum_{i=0}^{\rho-1} \sum_{(x_1,x_2)\in \gf(3)^2} \zeta_3^{\tr_{{3^2}/3}((x_1+x_2\beta^{\rho^2}) a(-1)^{i})+\tr_{{3^m}/3}( (x_1+x_2\beta^{h \rho}) b\beta^{h2i})},
\end{align}
where $\zeta_3=e^{2\pi \sqrt{-1}/3}$. Note that $\rho$ and $h$ are odd, we have $\rho\equiv h ~({\rm mod}~4)$ or $\rho\equiv -h~({\rm mod}~4)$. If $\rho\equiv h ~({\rm mod}~4)$, then $\beta^{h \rho}=\beta^{\rho^2}=(-1)^{\frac{\rho-1}2}\beta^{\rho}$. It follows from (\ref{EQ10}) that
\begin{align}\label{EQ11}
\wt(\widetilde{\bc}(a,b))&\geq \rho-\frac{1}{9}\sum_{i=0}^{\rho-1} \sum_{ x\in \gf(3^2)} \zeta_3^{\tr_{{3^2}/3}(x a(-1)^{i})+\tr_{{3^m}/3}( x b\beta^{h2i})}\notag \\
&=	 \rho-\frac{1}{9}\sum_{i=0}^{\rho-1} \sum_{ x\in \gf(3^2)} \zeta_3^{\tr_{{3^2}/3}( x (a (-1)^{i}+\tr_{{3^m}/{3^2}}(b \beta^{h2i} )  ) )}\notag \\
&=\left| \left \{ 0\leq i\leq \rho-1:\ a (-1)^{i}+\tr_{{3^m}/{3^2}}(b \beta^{h2i}) \neq 0 \right \}\ \right| \notag \\
&=\wt(\overline{\bc}(a, b) )\geq d(\overline{\C}(\rho)^{\bot}).
\end{align}
 If $\rho\equiv -h \pmod{4}$, then $\beta^{h \rho}=\beta^{3\rho^2}=(-1)^{\frac{\rho+1}2}\beta^{\rho}$. Note that 
\begin{align*}
\tr_{{3^2}/3}((x_1+x_2\beta^{\rho^2}) a (-1)^{i} )&=\tr_{{3^2}/3}((x_1^3+x_2^3\beta^{3\rho^2}) a^3 (-1)^{i} )\\
&=\tr_{{3^2}/3}((x_1+x_2\beta^{3\rho^2}) a^3 (-1)^{i} ).	
\end{align*}
It then follows from (\ref{EQ10}) that
\begin{align}\label{EQ12}
\wt(\widetilde{\bc}(a, b))&\geq \rho-\frac{1}{9}\sum_{i=0}^{\rho-1} \sum_{ x\in \gf(3^2)} \zeta_3^{\tr_{3^2/3}(x a^3(-1)^{i})+\tr_{3^m/3}( x b\beta^{h2i})}\notag \\
&=\wt(\overline{\bc}(a^3, b))\geq d(\overline{\C}(\rho)^{\bot}).
\end{align}
Combining Equations (\ref{EQ11}) and (\ref{EQ12}), we deduce that $d(\C(\rho))\geq d(\overline{\C}(\rho)^{\bot})$. The desired result follows.
\end{proof}

The lower bounds on $d(\C(\rho))$ and on $d(\C^\perp)$ documented in Theorems \ref{thm-am1} and \ref{thm-am2} are close 
to the square-root bound. 
The experimental data in Table \ref{tab-ding1} shows that the minimum distances of the codes $\C(\rho)$ (resp. $\C(\rho)^{\bot}$) and $\overline{\C}(\rho)$ (resp. $\overline{\C}(\rho)^{\bot}$) are very close to each other.

\begin{table*}
	\begin{center}
	\renewcommand\arraystretch{1.3}
\caption{The codes $\overline{\C}(\rho)$, $\C(\rho)$ and their dual}\label{tab-ding1}
\begin{tabular}{ccccc} \hline
 $\rho$ &  $\overline{\C}(\rho)$ &  $\overline{\C}(\rho)^{\bot}$ & $\C(\rho)$ &  $\C(\rho)^{\bot}$\\  \hline
 $5$ & $[5,2,4]$& $[5,3,3]$ &$[10,4,6]$ & $[10,6,4]$\\ \hline 
 $7$ & $[7,3,5]$& $[7,4,4]$ &$[14,6,6]$ & $[14,8,5]$\\ \hline 
 $17$ & $[17,8,8]$& $[17,9,7]$ &$[34,16,12]$ & $[34,18,10]$ \\ \hline 
 $19$ & $[19,9,10]$& $[19,10,9]$ & $[38,18,10]$ & $[38,20,9]$ \\ \hline 
 $29$ & $[29,14,12]$& $[29,15,11]$& $[58,28,18]$ & $[58,30,16]$ \\ \hline 
 $31$ & $[31,15,12]$& $[31,16,11]$ & $[62,30,14]$& $[62,32,12]$ \\ \hline 
 $43$ & $[43,21,16]$& $[43,22,15]$ & $[86,42,18]$&  $[86,44,16]$\\  \hline
\end{tabular}
\end{center}
\end{table*}

Note that the length of the code $\C(\rho)$ is even, the code $\C(\rho)$ is not known to be scalar-equivalent or permutation-equivalent to a ternary cyclic code. The negacyclic code $\C(\rho)$ and its dual are exceptionally good in general and are much better than the best cyclic codes with the same length and dimension. The authors are not aware of any family of ternary linear codes that outperforms this class of negacyclic codes $\C(\rho)$. The example codes below justify this claim. 

\begin{example}
Let $\rho=5$ and let $\beta$ be the primitive $4\rho$-th root of unity with $\beta^4+\beta^3+2\beta+1=0$. Then the negacylic code $\C(\rho)$ has parameters $[10,4,6]$, and $\C(\rho)^{\bot}$ has parameters $[10,6,4]$. These two negacyclic codes are distance-optimal \cite{Grassl}. The best ternary cyclic code of length $10$ and dimension $4$ has minimum distance $4$, and the best ternary cyclic code of length $10$ and dimension $6$ has minimum distance $2$ \cite{dingtang2014}.
\end{example}

\begin{example}
Let $\rho=7$ and let $\beta$ be the primitive $4\rho$-th root of unity with $\beta^6+2\beta^5+2\beta^3+2\beta+1=0$. Then the negacylic code $\C(\rho)$ has parameters $[14,6,6]$, and $\C(\rho)^{\bot}$ has parameters $[14,8,5]$. These two negacyclic codes are distance-optimal  \cite{Grassl}. The best ternary cyclic code of length $14$ and dimension $6$ has minimum distance $4$, and the best ternary cyclic code of length $14$ and dimension $8$ has minimum distance $2$ \cite{dingtang2014}.	
\end{example}

\begin{example}
Let $\rho=17$ and let $\beta$ be the primitive $4\rho$-th root of unity with $\beta^{16}+2\beta^{15}+2\beta^{12}+2\beta^{11}+2\beta^{10}+2\beta^6+\beta^5+2\beta^4+\beta+1=0$. Then the negacylic code $\C(\rho)$ has parameters $[34,16,12]$ and $\C(\rho)^{\bot}$ has parameters $[34,18,10]$. These two negacyclic codes are distance-optimal \cite{Grassl}. The best ternary cyclic code of length $34$ and dimension $16$ has minimum distance $4$, and the best ternary cyclic code of length $34$ and dimension $18$ has minimum distance $2$ \cite{dingtang2014}.	 
\end{example}

\begin{example}
Let $\rho=19$ and let $\beta$ be the primitive $4\rho$-th root of unity with 
$
\beta^{18} + 2\beta^{17} + 2\beta^{14} + \beta^{13} + 2\beta^{11} + \beta^{10} + 
2\beta^9 + \beta^8 + 2\beta^7 + \beta^5 + 2\beta^4 + 2\beta + 1=0$. 
Then the negacylic code $\C(\rho)$ has parameters $[38,18,10]$. The best ternary code known of length $38$ and dimension $18$ has minimum distance $12$ \cite{Grassl}. The best ternary cyclic code of length $38$ and dimension $18$ has minimum distance $4$ \cite{dingtang2014}. The dual code $\C(\rho)^{\bot}$ has parameters $[38,20,9]$. The best ternary code known of length $38$ and dimension $20$ has minimum distance $10$ \cite{Grassl}. The best ternary cyclic code of length $38$ and dimension $20$ has minimum distance $2$ \cite{dingtang2014}. 
\end{example}

\begin{example}
Let $\rho=29$ and let $\beta$ be the primitive $4\rho$-th root of unity with $\beta^{28}+2\beta^{27}+2\beta^{25}+\beta^{23}+2\beta^{22}+\beta^{21}+2\beta^{19}+2\beta^{17}+\beta^{15}+2\beta^{14}+2\beta^{13}+\beta^{11}+\beta^9+2\beta^7+2\beta^6+2\beta^5+\beta^3+\beta+1=0$. Then the negacylic code $\C(\rho)$ has parameters $[58,28,18]$, and has the best parameters known \cite{Grassl}. The best ternary cyclic code of length $58$ and dimension $28$ has minimum distance $4$ \cite{dingtang2014}. The dual code $\C(\rho)^{\bot}$ has parameters $[58,30,16]$, and has the best parameters known \cite{Grassl}. The best ternary cyclic code of length $58$ and dimension $30$ has minimum distance $2$ \cite{dingtang2014}. 
\end{example}

\begin{example}
Let $\rho=31$ and let $\beta$ be the primitive $4\rho$-th root of unity with $\beta^{30}+\beta^{29}+2\beta^{27}+\beta^{26}+2\beta^{25}+2\beta^{24}+\beta^{23}+2\beta^{22}+\beta^{21}+\beta^{20}+2\beta^{19}+\beta^{18}+\beta^{17}+2\beta^{16}+\beta^{15}+2\beta^{14}+\beta^{13}+\beta^{12}+2\beta^{11}+\beta^{10}+\beta^9+2\beta^8+\beta^7+2\beta^6+2\beta^{5}+\beta^4+2\beta^3+\beta+1=0$. Then the negacylic code $\C(\rho)$ has parameters $[62,30,14]$. The best ternary code known of length $62$ and dimension $30$ has minimum distance $18$ \cite{Grassl}. The best ternary cyclic code of length $62$ and dimension $30$ has minimum distance $4$ \cite{dingtang2014}. The dual code $\C(\rho)^{\bot}$ has parameters $[62,32,12]$. The best ternary code known of length $62$ and dimension $32$ has minimum distance $16$ \cite{Grassl}. The best ternary cyclic code of length $62$ and dimension $32$ has minimum distance $2$ \cite{dingtang2014}. 
\end{example}

\section{The second family of ternary negacyclic codes \& their duals}\label{sec4}

 Let $\ell \geq 2$ be an integer, and let $n$ be a positive divisor of $3^{\ell}+1$ and $2n> 3^{\lceil \ell/2\rceil}+1$. Let $\beta$ be a primitive $2n$-th root of unity, then $\beta^n=-1$. Let $\m_{\beta}(x)$ be the minimum polynomial of $\beta$ over $\gf(3)$. Let $\C(\ell)$ denote the ternary negacyclic code of length $n$ with check polynomial $\m_{\beta}(x)$. To settle the dimension of the negacyclic code $\C(\ell)$, we need the following lemma. 

\begin{lemma}\label{lem18}
Let notation be the same as before. Then $\ord_{2n}(3)=2\ell$.	
\end{lemma}

\begin{proof}  
It is clear that $\ord_{2n}(3)\mid 2\ell$.
\begin{enumerate}
\item If $\ord_{2n}(3) \leq  \frac{\ell}2$, we have $2n\leq 3^{\ell/2}-1$, a contradiction. 
\item If $\ord_{2n}(3) = \frac{2\ell}3$, then $3\mid \ell$, and $n$ divides $ \gcd\left(\frac{3^{2\ell/3}-1}2, 3^{\ell}+1 \right)$. Note that $$\gcd\left(\frac{3^{2\ell/3}-1}2, 3^{\ell}+1 \right)=3^{\ell/3}+1,$$
we have $n\leq 3^{\ell/3}+1$, a contradiction.
\item If $\ord_{2n}(3) =\ell $, we have $n$ divides $\gcd\left(\frac{3^{\ell}-1}2, 3^{\ell}+1 \right)$. Note that $$\gcd\left(\frac{3^{\ell}-1}2, 2\right)=\gcd(\ell,2),$$ 
we have $n\leq 2$, a contradiction.
\end{enumerate}
This completes the proof.	
\end{proof}

\begin{theorem}\label{thm19}
Let $n=\frac{3^{\ell}+1}2$, where $\ell \geq 2$ is an integer. Then the negacyclic code $\C(\ell)$ has parameters $[n, 2\ell, d\geq \frac{3^{\ell-1}+1}2]$, and the code $\C(\ell)^{\bot}$ has parameters $[n, n-2\ell,5]$.
\end{theorem}

\begin{proof}
	By Lemma \ref{lem18}, $\dim(\C(\ell))=\ord_{2n}(3)=2\ell$. Consequently, $\dim(\C(\ell)^{\bot})=n-2\ell$. It is easily verified that $C_1^{(3,2n)}=\left \{1,3,\cdots,3^{\ell-1}, 2n-3^{\ell-1}, 2n-3^{\ell-2},\cdots,2n-1 \right \}$. Let $$H=\left\{2i+1:\frac{3^{\ell-1}+1}2\leq i\leq 3^{\ell-1}-1 \right\}.$$
Then $H\subseteq \{1+2i:\ 0\leq i\leq n-1 \}\backslash C_1^{(3,2n)}$, and $\beta^i$ is a zero of $\C(\ell)$ for each $i\in H$. The desired lower bound then follows from Lemma \ref{lem1}. Notice that $-3,-1, 1, 3 \in C_1^{(3,2n)}$, from Lemma \ref{lem1}, $d(\C(\ell)^{\bot})\geq 5$. It follows from Lemma \ref{lem3} that $d(\C(\ell)^{\bot})\leq 5$. Therefore, $d(\C(\ell)^{\bot})=5$. This completes the proof.
\end{proof}

The code $\C(\ell)^{\bot}$ is a ternary negacyclic BCH code with designed distance $2$, and its parameters were studied in \cite{WSD22}. Theorem \ref{thm19} gives a lower bound of the minimum distance of the code $\C(\ell)$. When $\ell$ is even, $n=\frac{3^{\ell}+1}2$ is odd, then the ternary negacyclic code $\C(\ell)$ of length $n$ is scalar-equivalent to a ternary cyclic code of length $n$. In other words, Theorem \ref{thm19} can produce ternary cyclic codes with good parameters.

\begin{example}
	Let $\ell=3$, then $n=\frac{3^{\ell}+1}2=14$. Let $\beta$ be the primitive $2n$-th root of unity with $\beta^{6}+2\beta^5+2\beta^3+2\beta+1=0$. Then the negacylic code $\C(\ell)$ has parameters $[14,6,6]$ and $\C(\ell)^{\bot}$ has parameters $[14,8,5]$. These two negacyclic codes are distance-optimal \cite{Grassl}. The best ternary cyclic code of length $14$ and dimension $6$ has minimum distance $4$, and the best ternary cyclic code of length $14$ and dimension $8$ has minimum distance $2$ \cite{dingtang2014}.	 
\end{example}

\begin{example}
	Let $\ell=4$, then $n=\frac{3^{\ell}+1}2=41$. Let $\beta$ be the primitive $2n$-th root of unity with $\beta^{8}+2\beta^6+2\beta^5+2\beta^3+2\beta^2+1=0$. Then the negacylic code $\C(\ell)$ has parameters $[41,8,22]$, and has the best parameters known \cite{Grassl}. The dual code $\C(\ell)^{\bot}$ has parameters $[41,33,5]$, and is distance-optimal.
\end{example}
 
 \begin{example}
	Let $\ell=5$, then $n=\frac{3^{\ell}+1}2=122$. Let $\beta$ be the primitive $2n$-th root of unity with $\beta^{10}+2\beta^9+2\beta^8+2\beta^7+\beta^6+\beta^5+\beta^4+2\beta^3+2\beta^2+2\beta+1=0$. Then the negacylic code $\C(\ell)$ has parameters $[122,10,71]$. The best ternary code known of length $122$ and dimension $10$ has minimum distance $72$ \cite{Grassl}. The dual code $\C(\ell)^{\bot}$ has parameters $[122,112,5]$, and is distance-optimal.
\end{example}

\begin{theorem}\label{thm25}
Let $n=\frac{3^{\ell}+1}{4}$, where $\ell \geq 5$ is an odd integer. Then the negacyclic code $\C(\ell)$ has parameters $[n, 2\ell, d\geq \frac{3^{\ell-1}-1}4]$, and the code $\C(\ell)^{\bot}$ has parameters $[n, n-2\ell,5\leq d^{\bot}\leq 6]$.
\end{theorem}

\begin{proof}
	By Lemma \ref{lem18}, $\dim(\C(\ell))=\ord_{2n}(3)=2\ell$. Consequently, $\dim(\C(\ell)^{\bot})=n-2\ell$. It is easily verified that $C_1^{(3,2n)}=\left\{1,3,\cdots,3^{\ell-2},2n-3^{\ell-1}, 3^{\ell-1}, 2n-3^{\ell-2},\cdots,2n-1\right\}$. Let $$H=\left\{2i+1:\frac{3^{\ell-1}+3}4\leq i\leq \frac{3^{\ell-1}-3      }{2} \right\}.$$
Then $H\subseteq \{1+2i:\ 0\leq i\leq n-1\}\backslash C_1^{(3,2n)}$, and $\beta^i$ is a zero of $\C(\ell)$ for each $i\in H$. The desired lower bound then follows from Lemma \ref{lem1}. Notice that $-3,-1, 1, 3 \in C_1^{(3,2n)}$, from Lemma \ref{lem1}, $d(\C(\ell)^{\bot})\geq 5$. It follows from Lemma \ref{lem2} that $d(\C(\ell)^{\bot})\leq 6$. This completes the proof.
\end{proof}

Note that $3^{\ell}+1\equiv 4~({\rm mod}~8)$ for $\ell$ being odd, then $n=\frac{3^{\ell}+1}{4}$ is odd. In this case, the ternary negacyclic code $\C(\ell)$ of length $n$ is scalar-equivalent to a ternary cyclic code of length $n$. Therefore, Theorem \ref{thm25} can produce ternary cyclic codes with good parameters.

 \begin{example}
	Let $\ell=5$, then $n=\frac{3^{\ell}+1}4=61$. Let $\beta$ be the primitive $2n$-th root of unity with $\beta^{10}+\beta^8+\beta^7+\beta^5+\beta^3+\beta^2+1=0$. Then the negacylic code $\C(\ell)$ has parameters $[61,10,31]$. The best ternary code known of length $61$ and dimension $10$ has minimum distance $32$ \cite{Grassl}. The dual code $\C(\ell)^{\bot}$ has parameters $[61,51,5]$ and has the best parameters known \cite{Grassl}. 
	\end{example}

\begin{theorem}\label{thm-a171}
Let $n=3^{\ell}+1$, where $\ell \geq 2$ is an integer. Then the negacyclic code $\C(\ell)$ has parameters $[n, 2\ell, d\geq \frac{3^{\ell}+3}2]$, and the code $\C(\ell)^{\bot}$ has parameters $[n, n-2\ell, d^{\bot}]$, where 
\begin{align*}
	d^{\bot}=\begin{cases}
	3~&\ell~{\rm is~odd},\\
	4~&\ell~{\rm is~even}.
\end{cases}
\end{align*}
\end{theorem}
\begin{proof}
	By Lemma \ref{lem18}, $\dim(\C(\ell))=\ord_{2n}(3)=2\ell$. Consequently, $\dim(\C(\ell)^{\bot})=n-2\ell$. It is easily verified that $3^{\ell+i}\equiv 3^{\ell}-(3^i-1)~({\rm mod}~2n)$ for any $1\leq i\leq \ell-1$. Therefore,
	$$C_1^{(3,2n)}=\left\{1,3,\cdots,3^{\ell-1},3^{\ell}-(3^{\ell-1}-1),\cdots, 3^{\ell}-2, 3^{\ell}\right\}.$$ Let $$H=\left\{2i+1:\frac{3^{\ell}+1}2\leq i\leq 3^{\ell} \right\}.$$
Then $H\subseteq \{1+2i:\ 0\leq i\leq n-1\}\backslash C_1^{(3,2n)}$, and $\beta^i$ is a zero of $\C(\ell)$ for each $i\in H$. The desired lower bound then follows from Lemma \ref{lem1}.

Now we consider the minimum distance of $\C(\ell)^{\bot}$. It follows from Lemma \ref{lem2} that $d(\C(\ell)^{\bot})\leq 4$. Note that $1,3\in C_1^{(3,2n)}$, by Lemma \ref{lem1}, $d(\C(\ell)^{\bot})\geq 3$. We now consider the following two cases.
\begin{enumerate}
\item If $m$ is odd, then $4\mid n$. Let $\eta=\beta^{-\frac{n}{4}}$, then $\eta \in \gf(3^2)\backslash \gf(3)$. Consequently, there exist $a_0,a_1\in \gf(3)$ such that $\m_{\eta}(x)=x^2+a_1x+a_0$. It follows that $x^{\frac{n}2}+a_1x^{\frac{n}4}+a_0\in \C(\ell)^{\bot}$. Therefore, $d(\C(\ell)^{\bot})=3$.
\item If $m$ is even, then $\frac{n}2$ is odd. If $d(\C(\ell)^{\bot})=3$, then there exist $a_0,a_1\in \gf(3)^*$ and $1\leq i_0<i_1\leq n-1$ such that 
\begin{equation}\label{EQ14}
	a_1 \beta^{-i_1}+a_0 \beta^{-i_0}=1.
\end{equation}
Raising both sides of (\ref{EQ14}) to the $(3^m+1)$-th power, we have
\begin{equation}\label{EQ-15}
	(-1)^{i_1}+a_1a_0 [ (-1)^{i_1}\beta^{i_1-i_0}+(-1)^{i_0}\beta^{i_0-i_1}]+(-1)^{i_0}=1.
\end{equation}
It follows that $(-1)^{i_1}\beta^{i_1-i_0}+(-1)^{i_0}\beta^{i_0-i_1}\in \gf(3)$, which is equivalent to 
\begin{align}
	&[(-1)^{i_1}\beta^{i_1-i_0}+(-1)^{i_0}\beta^{i_0-i_1}]^3=(-1)^{i_1}\beta^{i_1-i_0}+(-1)^{i_0}\beta^{i_0-i_1} \notag \\
	\Leftrightarrow & (-1)^{i_0} \beta^{3(i_0-i_1)}( \beta^{2(i_1-i_0)}-1)[(-1)^{(i_1-i_0)}\beta^{4(i_1-i_0)}-1]=0.  \label{EQ-16}         
\end{align}
Note that $(-1)^{i_0}\beta^{3(i_0-i_1)}\neq 0$, it follows from (\ref{EQ-16}) that $\beta^{2(i_1-i_0)}=1$ or $\beta^{4(i_1-i_0)}=(-1)^{i_1-i_0}$.
\begin{enumerate}
\item If $\beta^{2(i_1-i_0)}=1$, we have $n \mid (i_1-i_0)$, which contradicts the fact that $0<i_1-i_0<n$.
\item If $\beta^{4(i_1-i_0)}=(-1)^{i_1-i_0}$, we have $n\mid 4(i_1-i_0)$. Since $\frac{n}2$ is odd, $\frac{n}2\mid (i_1-i_0)$. Notice that $0<i_1-i_0<n$, then $i_1-i_0=\frac{n}2$. Consequently, $\beta^{4(i_1-i_0)}=\beta^{2n}=1= (-1)^{i_1-i_0}=-1$, a contradiction.
\end{enumerate}
Therefore, $d(\C(\ell)^{\bot})=4$. 
\end{enumerate}
This completes the proof.
\end{proof}

\begin{example}\label{exam-a171}
Let $\ell=2$, then $n=3^{\ell}+1=10$. Let $\beta$ be the primitive $2n$-th root of unity with $\beta^4+\beta^3+2\beta+1=0$. Then the negacylic code $\C(\ell)$ has parameters $[10,4,6]$ and $\C(\ell)^{\bot}$ has parameters $[10,6,4]$. These two negacyclic codes 
are distance-optimal \cite{Grassl}. The best ternary cyclic code of length $10$ and dimension $4$ has minimum distance $4$, and the best ternary cyclic code of length $10$ and dimension $6$ has minimum distance $2$ \cite{dingtang2014}.	 
	
\end{example}

\begin{example}\label{exam-a172}
Let $\ell=3$, then $n=3^{\ell}+1=28$. Let $\beta$ be the primitive $2n$-th root of unity with $\beta^6+2\beta^5+2\beta+2=0$. Then the negacylic code $\C(\ell)$ has parameters $[28,6,15]$ and is distance-optimal \cite{Grassl}. The code $\C(\ell)^{\bot}$ has parameters $[28,22,3]$ and is distance-almost-optimal.	The best ternary cyclic code of length $28$ and dimension $6$ has minimum distance $12$, and the best ternary cyclic code of length $28$ and dimension $22$ has minimum distance $2$ \cite{dingtang2014}.	 
\end{example}

\begin{example}\label{exam-a173}
Let $\ell=4$, then $n=3^{\ell}+1=82$. Let $\beta$ be a primitive $2n$-th root of unity with $\beta^8+2\beta^7+2\beta^6+\beta^5+\beta^4+2\beta^3+2\beta^2+\beta+1=0$. Then the negacylic code $\C(\ell)$ has parameters $[82,8,48]$. The best ternary code known of length $82$ and dimension $8$ has minimum distance $49$ \cite{Grassl}. The code $\C(\ell)^{\bot}$ has parameters $[82,74,4]$ and is distance-optimal.	 
\end{example}

\section{The third family of ternary negacyclic codes and their duals}\label{sec5}

Let $m \geq 3$ be an integer, and let $n$ be a positive divisor of $\frac{3^{m}-1}{2}$ and $n>\frac{3^{ \lfloor m/2\rfloor}+1}2$. Let $\beta$ be a primitive $2n$-th root of unity, then $\beta^n=-1$. Let $\m_{\beta}(x)$ be the minimum polynomial of $\beta$ over $\gf(3)$. Let $\C(m)$ denote the ternary negacyclic code of length $n$ with check polynomial ${\rm lcm}(\m_{\beta}(x),\m_{\beta^{2n-1}}(x))$, where lcm denotes the least common multiple of the minimal polynomials. To settle the dimension of the negacyclic code $\C(m)$, we need the following lemma. 

\begin{lemma}\label{lem26}
Let notation be the same as before. Then $\ord_{2n}(3)=m$ and 
$${\rm lcm}(\m_{\beta}(x),\m_{\beta^{2n-1}}(x))=\m_{\beta}(x)\m_{\beta^{2n-1}}(x).$$ 
\end{lemma}

\begin{proof}  
It is clear that $\ord_{2n}(3)$ divides $m$.
\begin{enumerate}
\item If $\ord_{2n}(3)<m$ and $m$ is even, we have $2n\leq 3^{m/2}-1$, a contradiction. 
\item If $\ord_{2n}(3)<m$ and $m$ is odd, then $\ord_{2n}(3)\leq \frac{m}{3}\leq \frac{m-1}2$. Consequently, $2n\leq 3^{(m-1)/2}-1$, a contradiction.
\end{enumerate}
Therefore, $\ord_{2n}(3)=m$. It is easy to see that ${\rm lcm}(\m_{\beta}(x),\m_{\beta^{2n-1}}(x))=\m_{\beta}(x)\m_{\beta^{2n-1}}(x)$ if and only if $2n-1\notin C_1^{(3,2n)}$. Suppose $2n-1\in C_1^{(3,2n)}$, then there is $0\leq \ell\leq m-1$ such that $$3^{\ell}\equiv -1~({\rm mod}~2n).$$
It follows that $2n\mid \gcd(3^{\ell}+1,3^m-1)$. Note that 
\begin{align*}
	\gcd(3^{\ell}+1, 3^m-1)=\begin{cases}
	2~ &{\rm if}~\frac{m}{\gcd(\ell,m)}~{\rm is~odd},\\
	3^{\gcd(\ell, m)}+1~&{\rm if}~\frac{m}{\gcd(\ell,m)}~{\rm is~even}.
\end{cases}
\end{align*}
Since $2n>2$, from $2n\mid \gcd(3^{\ell}+1,3^m-1)$, we have $m$ is even and $\gcd(\ell, m)\leq \frac{m}2$. Consequently, $2n\leq 3^{m/2}+1$, a contradiction. This completes the proof.
\end{proof}

\begin{theorem}\label{thm32}
Let $n=\frac{3^{m}-1}2$, where $m \geq 3$ is an integer. Then the negacyclic code $\C(m)$ has parameters $[n, 2m, d\geq \frac{3^{m-1}-1}2]$, and the code $\C(m)^{\bot}$ has parameters $[n, n-2m,5]$.
\end{theorem}

\begin{proof}
It is easily checked that $|C_1^{(3,2n)}|=|C_{2n-1}^{(3,2n)}|$. By Lemma \ref{lem26}, $|C_1^{(3,2n)}|=|C_{2n-1}^{(3,2n)}|=m$. It follows that $\dim(\C(m))=2m$ and $\dim(\C(m)^{\bot})=n-2m$. On one hand, 
\begin{align*}
C_1^{(3,2n)}\cup C_{2n-1}^{(3,2n)} =\left\{1, 3, \cdots,3^{m-1} \right\} \cup \left\{2n-1, 2n-3, \cdots, 2n-3^{m-1}\right\}.
\end{align*}
Let
$$ H=\left\{2i+1:\frac{3^{m-1}+1}2\leq i\leq 3^{m-1}-2 \right\},$$
then $H\subseteq \{1+2i:\ 0\leq i\leq n-1\}\backslash (C_1^{(3,2n)}\cup C_{2n-1}^{(3,2n)})$, and $\beta^i$ is a zero of $\C(m)$ for each $i\in H$. The desired lower bound then follows from Lemma \ref{lem1}. On the other hand, $-3,-1, 1, 3 \in C_1^{(3,2n)}\cup C_{2n-1}^{(3,2n)}$, from Lemma \ref{lem1}, $d(\C(m)^{\bot})\geq 5$. It follows from Lemma \ref{lem3} that $d(\C(m)^{\bot})\leq 5$. This completes the proof.	
\end{proof}

The negacyclic code $\C(m)^{\bot}$ was studied in \cite{DDR11}. Theorem \ref{thm32} gives a lower bound on the minimum distance of the code $\C(m)$.  When $m$ is odd, $n=\frac{3^{m}-1}2$ is odd, then the ternary negacyclic code $\C(m)$ of length $n$ is scalar-equivalent to a ternary cyclic code of length $n$. In other words, Theorem \ref{thm32} can produce ternary cyclic codes with good parameters.

\begin{example}\label{exam-a176}
	Let $m=3$, then $n=\frac{3^m-1}{2}=13$. Let $\beta$ be the primitive $2n$-th root of unity with $\beta^3+2\beta+1=0$. Then the negacyclic code $\C(m)$ has parameters $[13,6,6]$ and $\C(m)^{\bot}$ has parameters $[13,7,5]$. These two negacyclic codes 
	are distance-optimal \cite{Grassl}.
\end{example}

\begin{example}\label{exam-a177}
	Let $m=4$, then $n=\frac{3^m-1}{2}=40$. Let $\beta$ be the primitive $2n$-th root of unity with $\beta^4+2\beta^3+2=0$. Then the negacyclic code $\C(m)$ has parameters $[40,8,21]$ and has the best parameters known \cite{Grassl}. The code $\C(m)^{\bot}$ has parameters $[40,32,5]$, and is distance-optimal. The best ternary cyclic code of length $40$ and dimension $8$ has minimum distance $20$, and the best ternary cyclic code of length $40$ and dimension $32$ has minimum distance $4$ \cite{dingtang2014}.	 
\end{example}

\begin{example}
	Let $m=5$, then $n=\frac{3^m-1}{2}=121$. Let $\beta$ be the primitive $2n$-th root of unity with $\beta^5+2\beta+1=0$. Then the negacyclic code $\C(m)$ has parameters $[121,10,71]$. The best ternary code known of length $121$ and dimension $10$ has minimum distance $72$ \cite{Grassl}. The code $\C(m)^{\bot}$ has parameters $[121,111,5]$, and is distance-optimal.
\end{example}

\begin{theorem}\label{thm36}
Let $n=\frac{3^{m}-1}{4}$, where $m \geq 4$ is an even integer. Then the negacyclic code $\C(m)$ has parameters $[n, 2m, d\geq \frac{3^{m-1}+1}4]$, and the code $\C(m)^{\bot}$ has parameters $[n, n-2m,5\leq d^{\bot}\leq 6]$, and $d^{\bot}=5$ if $m\equiv 0~({\rm mod}~4)$.
\end{theorem}

\begin{proof}
By Lemma \ref{lem26}, $\deg({\rm lcm}(\m_{\beta}(x), \m_{\beta^{2n-1}}(x)) )=2m$. Then $\dim(\C(m))=2m$ and $$\dim(\C(m)^{\bot})=n-2m.$$ 
It is easily checked that
$$C_1^{(3,2n)}\cup C_{2n-1}^{(3,2n)}=\left\{1,3, \cdots,3^{m-1} \right\}\cup \left\{2n-1, 2n-3,\cdots, 2n-3^{m-1} \right\} .$$
Let
$$ H=\left\{2i+1:\frac{3^{m-1}+1}4\leq i\leq \frac{3^{m-1}-3}2 \right\},$$
then $H\subseteq \{1+2i: \ 0\leq i\leq n-1\}\backslash (C_1^{(3,2n)}\cup C_{2n-1}^{(3,2n)})$, and $\beta^i$ is a zero of $\C(m)$ for each $i\in H$. The desired lower bound then follows from Lemma \ref{lem1}.

Notice that $-3,-1,1,3\in C_1^{(3,2n)}\cup C_{2n-1}^{(3,2n)}$, by Lemma \ref{lem1}, $d(\C(m)^{\bot})\geq 5$. It follows from Lemma \ref{lem2} that $d(\C(m)^{\bot})\leq 6$. If $m\equiv 0~({\rm mod}~4)$, $5$ divides $n$. Let $\eta=\beta^{\frac{n}5}$, then $\ord(\eta)=10$ and $\eta \in \gf(3^4)\backslash \gf(3^2)$. Then $\m_{\eta}(x)=\m_{\eta^{-1}}(x)$ and $\deg(\m_{\eta}(x))=4$. It follows that there exist $a_0,a_1,a_2,a_3\in \gf(3)$ such that $$\eta^4+c_3\eta^3+c_2\eta^2+c_1\eta+c_0=\eta^{-4}+c_3\eta^{-3}+c_2\eta^{-2}+c_1\eta^{-1}+c_0=0.$$
Therefore, $c(x)=x^{\frac{4n}{5}}+c_3x^{\frac{3n}{5}}+c_2 x^{\frac{2n}{5}}+c_1x^{\frac{n}{5}}+c_0\in \C(m)^{\bot}$, it deduces that $d(\C(m)^{\bot})\leq 5$. The desired conclusion then follows.
\end{proof}

\begin{example}\label{exam-a178}
	Let $m=4$, then $n=\frac{3^m-1}{4}=20$. Let $\beta$ be the primitive $2n$-th root of unity with $\beta^4+2\beta^3+\beta^2+1=0$. Then the negacyclic code $\C(m)$ has parameters $[20,8,8]$. The best ternary code known of length $20$ and dimension $8$ has minimum distance $9$ \cite{Grassl}. The code $\C(m)^{\bot}$ has parameters $[20,12,5]$. The best ternary code known of length $20$ and dimension $12$ has minimum distance $6$ \cite{Grassl}. The best ternary cyclic code of length $20$ and dimension $8$ has minimum distance $8$, and the best ternary cyclic code of length $20$ and dimension $12$ has minimum distance $4$ \cite{dingtang2014}. 
\end{example}

\begin{example}\label{exam-a179}
	Let $m=6$, then $n=\frac{3^m-1}{4}=182$. Let $\beta$ be the primitive $2n$-th root of unity with $\beta^6+\beta^5+2\beta^3+1=0$. Then the negacyclic code $\C(m)$ has parameters $[182,12,104]$. The best ternary code known of length $182$ and dimension $12$ has minimum distance $105$ \cite{Grassl}. The code $\C(m)^{\bot}$ has parameters $[182,170,5]$, and has the best parameters known \cite{Grassl}.  
\end{example}

\section{The fourth family of ternary negacyclic codes \& their duals}\label{sec6} 

Let $m\geq 2$ be an integer. Let $n=\frac{3^m-1}2$ and $N=2n=3^m-1$. For any integer $i$, let $i \bmod N$ denote the unique integer $s$ such that $0 \leq s \leq N-1$ and $i-s$ is divisible by $N$ throughout this section.  
For any $i$ with $0\leq i\leq N-1$, we have the following $3$-adic expansion $i=\sum_{j=0}^{m-1}i_j 3^j$, where $0\leq i_j\leq 2$. The $3$-weight $\wt_3(i)$ of $i$ is defined by $\wt_3(i)=\sum_{j=0}^{m-1}i_j$. It is easy to see that $\wt_3(i)$ is a constant on each cyclotomic coset $C_j^{(3, N)}$ and $\wt_3(i)\equiv i~({\rm mod}~2)$. Let $\Gamma_{(3,N)}$ be the set of $3$-cyclotomic coset leaders modulo $N$ and let $\Gamma_{(3,N)}^{(1)}=\{i: \ i\in \Gamma_{(3,N)}, \ i\equiv 1~({\rm mod}~2) \}$. Let $\beta$ be a primitive element of $\gf(3^m)$, then
\begin{align*}
x^n+1&=\prod_{i\in \Gamma_{(3,N)}^{(1)} }\m_{\beta^i}(x) \notag\\
&=\prod_{i\in \Gamma_{(3,N)}^{(1)}, \atop \wt_3(i)\equiv 1~({\rm mod}~4)} \m_{\beta^i}(x) \prod_{i\in \Gamma_{(3,N)}^{(1)}, \atop \wt_3(i)\equiv 3~({\rm mod}~4)}\m_{\beta^i}(x).
\end{align*}
For each $j\in \{1,3\}$, define 
$$g_{(j,m)}(x)=\prod_{i\in \Gamma_{(3,N)}^{(1)}, \atop \wt_3(i)\equiv j~({\rm mod}~4)} \m_{\beta^i}(x).$$
Then $g_{(j,m)}(x)$ is a polynomial over $\gf(3)$. 
Let $\C_{(j, m)}$ be the ternary negacyclic code of length $\frac{3^m-1}2$ with generator polynomial $g_{(j, m)}(x)$.

For any positive integer $m$, let
$$S_j(m)=\left \{ (i_0,i_1,\ldots, i_{m-1})\in \{0,1,2 \}^{m}:\ i_0+i_1+\cdots+i_{m-1}\equiv j~({\rm mod}~4)\right\},  $$
where $j\in \{0,1,2,3\}$. It is easy to verify that 
\begin{align*}
 \bigcup_{i\in \Gamma_{(3,N)}^{(1)}, \atop \wt_3(i)\equiv j~({\rm mod}~4) } C_i^{(3,N)}=\left\{1\leq i\leq N-1: \ \wt_3(i)\equiv j~({\rm mod}~4) \right\},
\end{align*}
and $\deg( g_{(j, m)}(x))=\arrowvert S_j(m)\arrowvert$ for $j=\{1,3\}$. To settle the dimension of the negacyclic code $\C_{(j, m)}$, we need the following lemma.

\begin{lemma}\label{Lemm38}
	Let notation be the same as before. The following hold.
	\begin{enumerate}
	\item If $m\geq 2$ is even, then $\arrowvert S_0(m)\arrowvert=  \frac{3^m+1+2(-1)^{\frac{m}2}}{4}$, $\arrowvert S_1(m)\arrowvert=\arrowvert S_3(m)\arrowvert=\frac{3^m-1}{4}$,  and  
	$$\arrowvert S_2(m)\arrowvert=\frac{3^m+1+2(-1)^{\frac{m-2}2}}{4}.$$ 
	\item If $m\geq 3$ is odd, then $\arrowvert S_0(m)\arrowvert=\arrowvert S_2(m)\arrowvert=\frac{3^m+1}{4}$, $\arrowvert S_1(m)\arrowvert=\frac{3^m-1+2(-1)^{\frac{m-1}2}}{4}$ and 
	$$\arrowvert S_3(m)\arrowvert=\frac{3^m-1+2(-1)^{\frac{m+1}2}}{4}.$$
	\end{enumerate}
	
\end{lemma}

\begin{proof}
It is easy to verify that 
\begin{align*}
\arrowvert S_1(m)\arrowvert+\arrowvert S_3(m)\arrowvert&=\arrowvert \{ (i_0,i_1,\ldots, i_{m-1})\in \{0,1,2 \}^{m}:\ i_0+i_1+\cdots+i _{m-1}\equiv 1~({\rm mod}~2)\} \arrowvert \\
&=\arrowvert \{ (i_0,i_1,\ldots, i_{m-1})\in \{0,1,2 \}^{m}:\ |\{0 \leq j \leq m-1: i_j=1\}| \mbox{ is odd} \} \arrowvert \\
&=\sum_{i {\rm ~ is ~odd}}\binom{m}{i} 2^{m-i}.
\end{align*}
Clearly,
\begin{align*}
(2+1)^m&=\sum_{i {\rm ~ is ~odd}}\binom{m}{i} 2^{m-i}+\sum_{i {\rm ~ is ~even}}\binom{m}{i} 2^{m-i},\\
(2-1)^m&=\sum_{i {\rm ~ is ~odd}}(-1)^i\binom{m}{i} 2^{m-i}+\sum_{i {\rm ~ is ~even}}(-1)^i\binom{m}{i} 2^{m-i}.
\end{align*}
Therefore,
\begin{equation}\label{EQ15}
\arrowvert S_1(m)\arrowvert+\arrowvert S_3(m)\arrowvert=\sum_{i {\rm ~ is ~odd}}\binom{m}{i} 2^{m-i}=\frac{3^m-1}2.
\end{equation}
Consequently, 
\begin{align}\label{EQ16}
\arrowvert S_0(m)\arrowvert+\arrowvert S_2(m)\arrowvert&=\arrowvert \{ (i_0,i_1,\ldots, i_{m-1})\in \{0,1,2 \}^{m}:\ i_0+i_1+\cdots+i_{m-1}\equiv 0~({\rm mod}~2)\} \arrowvert \notag \\
&= \frac{3^m+1}2.
\end{align}

For any $i=(i_0,i_1,\ldots, i_{m-1})\in \{0,1,2 \}^{m}$, define $i^c=(2-i_0,2-i_1,\ldots, 2-i_{m-1})$. 
We consider the following cases.  

\begin{enumerate}
\item Suppose that $m\geq 2$ is even. It is easy to verify that $i\in S_1(m)$ if and only if $i^c\in S_3(m)$. Then $$\psi_1:\  S_1(m)	\longrightarrow  S_3(m),\ i \longmapsto i^c$$ is a bijection. Therefore, $\arrowvert S_1(m)\arrowvert=\arrowvert S_3(m)\arrowvert$. It follows from (\ref{EQ15}) that 
$$\arrowvert S_1(m)\arrowvert=\arrowvert S_3(m)\arrowvert=\frac{3^m-1}{4}.$$
Clearly, $\arrowvert S_0(2)\arrowvert=2$. Suppose $m\geq 4$, then
\begin{align*}
	\arrowvert S_0(m)\arrowvert &=\arrowvert S_0(m-1)\arrowvert +\arrowvert S_2(m-1)\arrowvert +\arrowvert S_3(m-1)\arrowvert  \notag\\
	&=\frac{3^{m-1}+1}2+ \arrowvert S_3(m-1)\arrowvert \notag \\
	&=\frac{3^{m-1}+1}2+\arrowvert S_1(m-2)\arrowvert+ \arrowvert S_2(m-2)\arrowvert+ \arrowvert S_3(m-2)\arrowvert  \notag \\ 
	&=\frac{3^{m-1}+1}2+\frac{3^{m-2}-1}2+\frac{3^{m-2}+1}2-\arrowvert S_0(m-2)\arrowvert  \notag \\
	&=\frac{5\cdot 3^{m-2}+1}2-\arrowvert S_0(m-2)\arrowvert.
\end{align*}
Using recursion, we deduce that  
\begin{align}
	\arrowvert S_0(m)\arrowvert &=\sum_{i=1}^{\frac{m-2}2}(-1)^{i-1}\left(\frac{5\cdot 3^{m-2i}+1}2 \right)+(-1)^{\frac{m-2}2}\arrowvert S_0(2)\arrowvert \notag \\
	&= \frac{3^m+1+2(-1)^{\frac{m}2}}{4}.\label{EQ17}
\end{align}
It follows from (\ref{EQ16}) and (\ref{EQ17}) that 
$$\arrowvert S_2(m)\arrowvert=\frac{3^m+1+2(-1)^{\frac{m-2}2}}{4}. $$
\item Suppose that $m\geq 3$ is odd. It is easy to check that $i\in S_0(m)$ if and only if $i^c\in S_2(m)$. Then $$\psi_2:\  S_0(m)	\longrightarrow  S_2(m),\ i \longmapsto i^c$$ is a bijection. Therefore, $\arrowvert S_0(m) \arrowvert=\arrowvert S_2(m)\arrowvert$. It follows from (\ref{EQ16}) that 
$$\arrowvert S_0(m)\arrowvert=\arrowvert S_2(m)\arrowvert=\frac{3^m+1}{4}.$$
Note that 
\begin{align}
\arrowvert S_1(m)\arrowvert &=\arrowvert S_0(m-1)\arrowvert+\arrowvert S_1(m-1)\arrowvert  +\arrowvert S_3(m-1)\arrowvert  \notag\\
	&=\arrowvert S_0(m-1)\arrowvert+\frac{3^{m-1}-1}2 \notag\\
	&=\frac{3^{m-1}+1+2(-1)^{\frac{m-1}2}}{4}+\frac{3^{m-1}-1}2 \notag\\
	&=\frac{3^m-1+2(-1)^{\frac{m-1}2}}{4}\label{EQ18}
\end{align}
It follows from (\ref{EQ15}) and (\ref{EQ18}) that
$$\arrowvert S_3(m)\arrowvert=\frac{3^m-1+2(-1)^{\frac{m+1}2}}{4}.$$
\end{enumerate}
This completes the proof.
\end{proof} 

To estimate the minimum distance of the negacyclic code $\C_{{(j, m)}}$, we need the following lemmas.

\begin{lemma}\label{lemm39}
Let $2\leq s\leq m$ be a positive integr. For any $0\leq i\leq 3^{s}-1$, we have $$\wt_3(3^{s}-1-i)=2s-\wt_3(i).$$	
\end{lemma}

\begin{proof}
Let the $3$-adic expansion of $i$ be $i=\sum_{j=0}^{s-1}i_j3^j$, then 
$$3^{s}-1-i=\sum_{j=0}^{s-1}(2-i_j)3^j.$$ 
It follows that $\wt_3(3^{s}-i)=2s-\wt_3(i)$. This completes the proof.
\end{proof}

\begin{lemma}\label{lemm40}
Let $m\geq 5$ and $m \equiv 1~({\rm mod}~4)$. Let $N=3^m-1$ and $v=\frac{3^{m}-1}{2}+3^{\frac{m-1}2}-1$. Then $\gcd(v,N)=1$ and 
$$\left \{v(1+2i)\bmod{N}: \ 0\leq i\leq \frac{3^{\frac{m-1}2}-1}{4}+2 \right \} \subseteq S_1(m).$$	
\end{lemma}

\begin{proof}
Since $m\equiv 1~({\rm mod}~4)$,  $\frac{N}2$ and $v$ are odd. It is easily checked that 
\begin{align*}
\gcd\left(v, N\right)&= \gcd\left( v,  \frac{N}2 \right)\\
&= \gcd\left( 3^{\frac{m-1}2}-1, \frac{N}2 \right)\\
&=\frac{1}2 \cdot \gcd\left( 3^{\frac{m-1}2}-1, 3^{m}-1\right)=1.
\end{align*}
We now prove the second conclusion. It is easily checked that 
\begin{align}\label{EEQQ17}
v(1+2i)&\equiv \frac{3^m-1}2+(3^{\frac{m-1}2}-1)(1+2i)~({\rm mod}~N) \notag\\
&\equiv 3^{m-1}+3^{\frac{m-1}2}\left(\frac{3^{\frac{m-1}2}-1}2+1+2i \right)+\frac{3^{\frac{m-1}2}-1}2-(1+2i)~({\rm mod}~N).
\end{align}
The rest of the proof can be divided into the following cases.
\begin{enumerate}
\item Suppose that $0\leq i\leq \frac{3^{\frac{m-1}2}-1}{4}-1$, then $1\leq 1+2i\leq \frac{3^{\frac{m-1}2}-1}2-1$. It follows from (\ref{EEQQ17}) that
\begin{equation}\label{EEQQ18}
	\wt_3(v(1+2i)\bmod{N})=1+\wt_3\left(\frac{3^{\frac{m-1}2}-1}2+1+2i \right)+\wt_3\left(\frac{3^{\frac{m-1}2}-1}2-1-2i \right).
\end{equation}
By Lemma \ref{lemm39},
\begin{equation}\label{EEQQ19}
	\wt_3\left(\frac{3^{\frac{m-1}2}-1}2+1+2i \right)=2\left(\frac{m-1}2\right)-\wt_3\left(\frac{3^{\frac{m-1}2}-1}2-1-2i \right).
\end{equation}
From (\ref{EEQQ18}) and (\ref{EEQQ19}), 
\begin{equation}\label{EEQQ20}
	\wt_3(v(1+2i)\bmod{N})=m.
\end{equation}
\item Suppose that $ i=\frac{3^{\frac{m-1}2}-1}{4}$, then $1+2i=\frac{3^{\frac{m-1}2}-1}2+1$. It follows from (\ref{EEQQ17}) that
$$v(1+2i)\equiv 3^{m-1}+2\sum_{j=0}^{m-2}3^j~({\rm mod}~N). $$
Therefore, 
\begin{equation}\label{EEQQ21}
	\wt_3(v(1+2i)\bmod{N})=2m-1.
\end{equation}
\item Suppose that $i=\frac{3^{\frac{m-1}2}-1}{4}+1$, then $1+2i=\frac{3^{\frac{m-1}2}-1}2+3$. It follows from (\ref{EEQQ17}) that
$$v(1+2i)\equiv 2\cdot3^{m-1}+3^{\frac{m-1}2}+2\sum_{j=1}^{\frac{m-3}2}3^j ~({\rm mod}~N). $$
Therefore, 
\begin{equation}\label{EEQQ22}
	\wt_3(v(1+2i)\bmod{N})=m.
\end{equation}

\item Suppose that $i=\frac{3^{\frac{m-1}2}-1}{4}+2$, then $1+2i=\frac{3^{\frac{m-1}2}-1}2+5$. It follows from (\ref{EEQQ17}) that
$$v(1+2i)\equiv 2\cdot3^{m-1}+3^{\frac{m+1}2}+2\sum_{j=2}^{\frac{m-3}2}3^j+3+1 ~({\rm mod}~N). $$
Therefore, 
\begin{equation}\label{EEQQ23}
	\wt_3(v(1+2i)\bmod{N})=m.
\end{equation}
\end{enumerate}
It follows from Equations (\ref{EEQQ20}), (\ref{EEQQ21}), (\ref{EEQQ22}) and (\ref{EEQQ23}) that $\wt_3(v(1+2i)\bmod {N} )\equiv 1~({\rm mod}~4)$ for $0\leq i\leq \frac{3^{\frac{m-1}2}-1}{4}+2$. This completes the proof.
\end{proof}

\begin{lemma}\label{lemm41}
Let $m\geq 5$ and $m \equiv 1~({\rm mod}~4)$. Let $N=3^m-1$ and $v=\frac{3^{m}-1}{2}-3^{\frac{m-1}2}-1$. Then $\gcd(v,N)=1$ and 
$$\left \{v(1+2i) \bmod{N} : \ 0\leq i\leq \frac{3^{\frac{m-1}2}-1}{4} \right \} \subseteq S_3(m).$$	
\end{lemma}

\begin{proof}
 	Since $m\equiv 1~({\rm mod}~4)$,  $\frac{N}2$ and $v$ are odd. It is easily checked that 
\begin{align*}
\gcd\left(v, N\right)&= \gcd\left( v,  \frac{N}2 \right)\\
&= \gcd\left( 3^{\frac{m-1}2}+1, \frac{N}2 \right)\\
&=\frac{1}2 \cdot \gcd\left( 3^{\frac{m-1}2}+1, 3^{m}-1\right)=1.
\end{align*}
We now prove the second conclusion. It is easily checked that 
\begin{align}\label{EEQQ24}
v(1+2i)&\equiv \frac{3^m-1}2-(3^{\frac{m-1}2}+1)(1+2i)~({\rm mod}~N) \notag\\
&\equiv 3^{m-1}+3^{\frac{m-1}2}\left(\frac{3^{\frac{m-1}2}-1}2-1-2i \right)+\frac{3^{\frac{m-1}2}-1}2-1-2i ~({\rm mod}~N).
\end{align}
The rest of the proof can be divided into the following cases.
\begin{enumerate}
\item Suppose that $0\leq i\leq \frac{3^{\frac{m-1}2}-1}4-1$, then $1\leq 1+2i\leq \frac{3^{\frac{m-1}2}-1}2-1$. It follows from (\ref{EEQQ24}) that
\begin{equation}\label{EEQQ25}
	\wt_3(v(1+2i)\bmod{N})=1+2\cdot \wt_3\left(\frac{3^{\frac{m-1}2}-1}2-1-2i \right).
\end{equation}
Note that $\frac{3^{\frac{m-1}2}-1}2-1-2i$ is odd, we have $\wt_3\left(\frac{3^{\frac{m-1}2}-1}2-1-2i \right)\equiv 1~({\rm mod}~2)$. It then follows from (\ref{EEQQ25}) that
\begin{equation}\label{EEQQ26}
	\wt_3(v(1+2i)\bmod{N})\equiv 3~({\rm mod}~4).
\end{equation}
\item Suppose that $i=\frac{3^{\frac{m-1}2}-1}4$, then $1+2i=\frac{3^{\frac{m-1}2}-1}2+1$. It follows from (\ref{EEQQ24}) that
$$v(1+2i)\equiv 2\sum_{\frac{m+1}2}^{m-2}3^i+3^{\frac{m-1}2}+2\sum_{j=0}^{\frac{m-3}2}3^j~({\rm mod}~N).$$
Therefore, 
\begin{equation}\label{EEQQ27}
	\wt_3(v(1+2i)\bmod{N})=2m-3\equiv 3~({\rm mod}~4).
\end{equation}
\end{enumerate}
It follows from Equations (\ref{EEQQ26}) and (\ref{EEQQ27}) that $\wt_3(v(1+2i)\bmod {N} )\equiv 3~({\rm mod}~4)$ for $0\leq i\leq \frac{3^{\frac{m-1}2}-1}{4}$. This completes the proof.
\end{proof}

\begin{lemma}\label{lemm42}
Let $m\geq 7$ and $m \equiv 3 ~({\rm mod}~4)$. Let $N=3^m-1$ and $v=\frac{3^{m}-1}{2}-3^{\frac{m-1}2}-1$. Then $\gcd(v,N)=1$ and 
$$\left \{v(1+2i)\bmod{N}: \ 0\leq i\leq \frac{3^{(m-1)/2}+1}{4}+1 \right \} \subseteq S_1(m).$$	
\end{lemma}

\begin{proof}
	By Lemma \ref{lemm41}, we have $\gcd(v,N)=1$. We now prove the second conclusion by distinguishing the following cases. 
	\begin{enumerate}
	\item Suppose that $0\leq i\leq \frac{3^{\frac{m-1}2}+1}{4}-1$, then $1\leq 1+2i\leq \frac{3^{\frac{m-1}2}-1}2$. It follows from (\ref{EEQQ24}) that
\begin{equation}\label{EEQQ28}
	\wt_3(v(1+2i)\bmod{N})=1+2\cdot \wt_3\left(\frac{3^{\frac{m-1}2}-1}2-1-2i \right).
\end{equation}
Note that $\frac{3^{\frac{m-1}2}-1}2-1-2i$ is even, we have $\wt_3\left(\frac{3^{\frac{m-1}2}-1}2-1-2i \right)\equiv 0~({\rm mod}~2)$. It then follows from (\ref{EEQQ28}) that
\begin{equation}\label{EEQQ29}
	\wt_3(v(1+2i)\bmod{N})\equiv 1~({\rm mod}~4).
\end{equation}
\item Suppose that $i=\frac{3^{\frac{m-1}2}+1}{4}$, then $1+2i=\frac{3^{\frac{m-1}2}-1}2+2$. It follows from (\ref{EEQQ24}) that
$$v(1+2i)\equiv 2\sum_{j=\frac{m+1}2}^{m-2}3^j+2\sum_{j=1}^{\frac{m-3}2}+1~({\rm mod}~N).$$
Therefore, 
\begin{equation}\label{EEQQ30}
\wt_3(v(1+2i)\bmod{N})=2m-5\equiv 1~({\rm mod}~4).	
\end{equation}
\item Suppose that $i=\frac{3^{\frac{m-1}2}+1}{4}+1$, then $1+2i=\frac{3^{\frac{m-1}2}-1}2+4$. It follows from (\ref{EEQQ24}) that
$$v(1+2i)\equiv 2\sum_{j=\frac{m+3}2}^{m-2}3^j+3^{\frac{m+1}2}+3^{\frac{m-1}2}+2\sum_{j=2}^{\frac{m-3}2}+3+2~({\rm mod}~N).$$
Therefore, 
\begin{equation}\label{EEQQ31}
\wt_3(v(1+2i)\bmod{N})=2m-5\equiv 1~({\rm mod}~4).	
\end{equation}
	\end{enumerate}
It follows from Equations (\ref{EEQQ29}), (\ref{EEQQ30}) and (\ref{EEQQ31}) that $\wt_3(v(1+2i)\bmod{N} )\equiv 1~({\rm mod}~4)$ for $0\leq i\leq \frac{3^{\frac{m-1}2}+1}{4}+1$. This completes the proof.	
\end{proof}

\begin{lemma}\label{lemm43}
Let $m\geq 7$ and $m \equiv 3 ~({\rm mod}~4)$. Let $N=3^m-1$ and $v=\frac{3^m-1}2+3^{\frac{m-1}2}-1$. Then $\gcd(v,N)=1$ and 
$$\left \{v(1+2i) \bmod{N}: \ 0\leq i\leq \frac{3^{\frac{m-1}2}+1}{4} \right \} \subseteq S_3(m).$$	
\end{lemma}

\begin{proof}
	By Lemma \ref{lemm40}, we have $\gcd(v,N)=1$ and 
	$$\wt_3(v(1+2i)\bmod{N} )=m$$ 
	for $0\leq i\leq \frac{3^{\frac{m-1}2}-3}{4}$. If $i=\frac{3^{\frac{m-1}2}+1}4$, then 
	$$v(1+2i)\equiv 2\cdot 3^{m-1}+2\sum_{j=1}^{\frac{m-3}2}3^j+1~({\rm mod}~N).$$
	It follows that $\wt_3(v(1+2i)\bmod{N}) =m$. Therefore, $\wt_3(v(1+2i)\bmod{N}) \equiv 3~({\rm mod}~4)$ for $0\leq i\leq \frac{3^{\frac{m-1}2}+1}{4}$. This completes the proof.
\end{proof}

Based on the foregoing lemmas, the main results of this section are given in the next theorem.

\begin{theorem}
Let $m\geq 5$ be odd. Then the negacyclic code $\C_{(1,m)}$ has parameters 
$$\left[\frac{3^m-1}2,\, \frac{3^m-1+2(-1)^{\frac{m+1}2}}4,\, d\geq \frac{3^{\frac{m-1}2}+2+(-1)^{\frac{m-1}2}}4+3 \right],$$
and $\C_{(1,m)}^{\bot}=\C_{(3,m)}$, which has parameters 
$$\left [\frac{3^m-1}2,\, \frac{3^m-1+2(-1)^{\frac{m-1}2}}4,\,  d^{\bot}\geq \frac{3^{\frac{m-1}2}-(-1)^{\frac{m-1}2}}4+2 \right].$$ 
\end{theorem}

\begin{proof}
	It is clear that $\dim(\C_{(1,m)})=n-|S_1(m)|=|S_3(m)|$. The desired dimension of $\C_{(1,m)}$ then follows from Lemma \ref{Lemm38}. We now prove the lower bound on the minimum distance of the code $\C_{(1,m)}$. We consider the 
	following two cases.
	\begin{enumerate}
		\item Suppose that $m\equiv 1~({\rm mod}~4)$. Let $v=\frac{3^{m}-1}{2}+3^{\frac{m-1}2}-1$. It follows from Lemma \ref{lemm40} that $\gcd(v,N)=1$. Let $\gamma=\beta^v$, then $\gamma^n=-1$. It follows again from Lemma \ref{lemm40} that the defining set of $\C_{(1,m)}$ with respect to $\gamma$ contains the set $\left \{1+2i: \ 0\leq i\leq \frac{3^{\frac{m-1}2}-1}{4}+2 \right \}$. The desired lower bound on $d$ then follows from the BCH bound on negacyclic codes.
		\item Suppose that $m\equiv 3~({\rm mod}~4)$. Let $v=\frac{3^{m}-1}{2}-3^{\frac{m-1}2}-1$. It follows from Lemma \ref{lemm42} that $\gcd(v,N)=1$. Let $\gamma=\beta^v$, then $\gamma^n=-1$. It follows again from Lemma \ref{lemm42} that the defining set of $\C_{(1,m)}$ with respect to $\gamma$ contains the set $\left \{1+2i : \ 0\leq i\leq \frac{3^{(m-1)/2}+1}{4}+1 \right \}$. The desired lower bound on $d$ then follows from the BCH bound on negacyclic codes.
	\end{enumerate}

	It follows from Lemma \ref{lemm39} that $\wt_3(N-i)=2m-\wt_{3}(i)$. Therefore, 
	 $$\wt_{3}(N-i)\equiv \wt_{3}(i)~({\rm mod}~4)$$
	 for odd $i$. It follows that $\widehat{g_{(j, m)}}(x)=g_{(j, m)}(x)$ for $j\in \{1,3\}$. By definition, $\C_{(1,m)}=( g_{(1,m)}(x))$ and 
	 $$\C_{(1,m)}^{\bot}=( \widehat{g_{(3, m)}}(x))=( g_{(3,m)}(x)).$$
	Hence, $\C_{(1,m)}^{\bot}=\C_{(3,m)}$. 
	
	The desired dimension of $\C_{(3,m)}$ then follows from Lemma \ref{Lemm38}. We now prove the lower bound on the minimum distance of the code $\C_{(3,m)}$. We consider the following two cases.
 \begin{enumerate}
		\item Suppose that $m\equiv 1~({\rm mod}~4)$. Let $v=\frac{3^{m}-1}{2}-3^{\frac{m-1}2}-1$. It follows from Lemma \ref{lemm41} that $\gcd(v,N)=1$. Let $\gamma=\beta^v$, then $\gamma^n=-1$. It follows again from Lemma \ref{lemm41} that the defining set of $\C_{(3,m)}$ with respect to $\gamma$ contains the set $\left \{1+2i: \ 0\leq i\leq \frac{3^{\frac{m-1}2}-1}{4} \right \}$. The desired lower bound on $d^{\bot}$ then follows from the BCH bound on negacyclic codes.
		\item Suppose that $m\equiv 3~({\rm mod}~4)$. Let $v=\frac{3^m-1}2+3^{\frac{m-1}2}-1$. It follows from Lemma \ref{lemm43} that $\gcd(v,N)=1$. Let $\gamma=\beta^v$, then $\gamma^n=-1$. It follows again from Lemma \ref{lemm43} that the defining set of $\C_{(3,m)}$ with respect to $\gamma$ contains the set $\left \{1+2i : \ 0\leq i\leq \frac{3^{(m-1)/2}+1}{4} \right \}$. The desired lower bound on $d^{\bot}$ then follows from the BCH bound on negacyclic codes.
	\end{enumerate}
	This completes the proof.
\end{proof}

Since $m$ is odd,  $n=\frac{3^m-1}2$ is odd. Then the ternary negacyclic code $\C_{(i,m)}$ of length $n$ is scalar-equivalent to a ternary cyclic code of length $n$. Studying these negacyclic codes are still valuable due to the following facts. Firstly, our experimental data shows  that this family of negacyclic codes have very good parameters in general and contain distance-optimal codes. 
For example, when $m=3$, the negacyclic code $\C_{(1,m)}$ has parameters $[13,7,5]$, and $\C_{(3,m)}$ has parameters $[13,6,6]$. These two negacyclic codes are distance-optimal.	
Secondly, the ternary cyclic codes that are scalar-equivalent to these negacyclic codes $\C_{(i, m)}$ have not been studied in the literature. 

Finally, we compare the family of ternary negacyclic codes $\C_{(1,m)}$ with  the family of  ternary projective Reed-Muller codes 
(see \cite{Sorensen}).  The parameters of the ternary projective Reed-Muller codes of length $13$ are given below:
$$[13,3,9],\ \ [13,6,6],\ \ [13, 10,3], \ \ [13,12,2].  $$
Notice that the negacyclic code $\C_{(1,3)}$ has parameters $[13,7,5]$. 
The family of codes $\C_{(1, m)}$ and the ternary projective Reed-Muller codes are different in general. 
 
\section{Summary and concluding remarks}\label{sec7} 
 
The main contributions of this paper are the constructions and analyses of several families of ternary negacyclic codes. These ternary negacyclic codes are very interesting in theory as they contain distance-optimal codes and codes with best known parameters (see the code examples presented in this paper). A summary of the main specific contributions of this paper goes as follows.  
\begin{enumerate}
	\item A family of ternary irreducible negacyclic codes $\C(\rho)$ with parameters $$[2\rho, \rho-1, d\geq \sqrt{\rho}+1]$$ was constructed in Section \ref{sec3} 
	(see Theorems \ref{thm-am1} and \ref{thm-am2}). 
	The dual code $\C(\rho)^{\bot}$ has parameters $[2\rho, \rho+1,d^{\bot} \geq \sqrt{\rho}]$.  	
	The authors are not aware of any family of ternary codes that can outperform this family of negacyclic codes in terms of the error-correcting capability when the length and dimension are fixed. 
	\item A family of ternary irreducible negacyclic codes $\C(\ell)$ with parameters $\left[\frac{3^{\ell}+1}2, 2\ell, d\geq \frac{3^{\ell-1}+1}2\right]$ was constructed in Section \ref{sec4}. Moreover, the dual code $\C(\ell)^{\bot}$ has parameters $$\left[\frac{3^{\ell}+1}2, \frac{3^{\ell}+1}2-2\ell, 5\right]$$ and is distance-optimal (see Theorem \ref{thm19}). 
	\item A family of ternary irreducible negacyclic codes $\C(\ell)$ with parameters $\left[\frac{3^{\ell}+1}4, 2\ell, d\geq \frac{3^{\ell-1}-1}4\right]$ was constructed in Section \ref{sec4}. Moreover, the dual code $\C(\ell)^{\bot}$ has parameters $$\left[\frac{3^{\ell}+1}4, \frac{3^{\ell}+1}4-2\ell, 5\leq d^{\bot}\leq 6\right]$$ and is distance-almost-optimal (see Theorem \ref{thm25}).
	\item A family of ternary irreducible negacyclic codes $\C(\ell)$ with parameters $\left[3^{\ell}+1, 2\ell, d\geq \frac{3^{\ell}+3}2\right]$ was constructed in Section \ref{sec4} (see Theorem \ref{thm-a171}). Examples \ref{exam-a171},   \ref{exam-a172}, 
	and  \ref{exam-a173} show that the code could be much 
	better than the best ternary cyclic code with the same length and dimension and could be distance-optimal.  In addition, 
	we have the following: 
	\begin{itemize}
	\item If $\ell$ is odd, 	the dual code $\C(\ell)^{\bot}$ has parameters $[3^{\ell}+1, 3^{\ell}+1-2\ell,3]$, and is distance-almost-optimal.
	\item If $\ell$ is even, the dual code $\C(\ell)^{\bot}$ has parameters $[3^{\ell}+1, 3^{\ell}+1-2\ell,4]$, and is distance-optimal.
	\end{itemize}
	\item A family of ternary negacyclic codes $\C(m)$ with parameters $\left[\frac{3^{m}-1}2, 2m, d\geq \frac{3^{m-1}-1}2\right]$ was constructed in Section \ref{sec5} (see Theorem \ref{thm32}). Examples \ref{exam-a176}  
	and  \ref{exam-a177} show that the code could be distance-optimal or could have the best known parameters.  
	 Moreover, the dual code $\C(m)^{\bot}$ has parameters $$\left[\frac{3^{m}-1}2, \frac{3^{m}-1}2-2m, 5\right]$$ and is 
	 distance-optimal.
	\item A family of ternary negacyclic codes $\C(m)$ with parameters $\left[\frac{3^{m}-1}4, 2m, d\geq \frac{3^{m-1}+1}4\right]$ was constructed in Section \ref{sec5} (see Theorem \ref{thm36}). Examples \ref{exam-a178}  
	and  \ref{exam-a179} show that the code could be much 
	better than the best ternary cyclic code with the same length and dimension. 
	Moreover, the dual code $\C(m)^{\bot}$ has parameters $$\left[\frac{3^{m}-1}4, \frac{3^{m}-1}4-2m, d^{\bot}\right],$$ where $5\leq d^\bot\leq 6$, and $d^{\bot}=5$ if $m\equiv 0\pmod{4}$,  and $\C(m)^{\bot}$ is distance-almost-optimal.
	\item A family of ternary negacyclic codes $\C_{(1,m)}$ with parameters $$\left[\frac{3^m-1}2,\frac{3^m-1+2(-1)^{\frac{m+1}2}}4, d\geq \frac{3^{\frac{m-1}2}+2+(-1)^{\frac{m-1}2}}4+3 \right] $$ 
	was constructed in Section \ref{sec6}. The dual code $\C_{(1,m)}^{\bot}$ has parameters $$\left [\frac{3^m-1}2,\frac{3^m-1+2(-1)^{\frac{m-1}2}}4, d^{\bot}\geq \frac{3^{\frac{m-1}2}-(-1)^{\frac{m-1}2}}4+2 \right],$$ 
\end{enumerate} 
The codes $\C_{(1,m)}$ and $\C_{(1,m)}^\perp$ have very good parameters in general. 

Some families of ternary negacyclic codes presented in this paper have odd lengths. So these ternary negacyclic codes are scalar-equivalent to some ternary cyclic codes. Therefore, this paper has produced some families of ternary cyclic codes with good parameters, which 
were not studied in the literature. Notice that most families of ternary negacyclic codes presented in this paper have even lengths, and they have a much better error-correcting capability compared with ternary cyclic codes with the same length and dimension.  


\end{document}